\documentclass[reprint,aps,prx,superscriptaddress,longbibliography]{revtex4-2}

\usepackage{amsmath,amssymb,amsfonts}
\usepackage{amsthm}

\usepackage{graphicx}
\usepackage{physics}
\usepackage{bm}
\usepackage{hyperref}
\usepackage{xcolor}
\usepackage{booktabs}
\usepackage{subcaption}
\usepackage{float}
\usepackage{algpseudocode}  % algorithmicx

\newcounter{algcnt}
\renewcommand{\thealgcnt}{\arabic{algcnt}}
\newcommand{\algcaption}[1]{%
	\refstepcounter{algcnt}%
	\par\vspace{2pt}\noindent
	\begingroup\small\textbf{Algorithm \thealgcnt.} #1\par\endgroup
	\vspace{4pt}\hrule\vspace{5pt}}
\usepackage{mathtools}
\usepackage{tikz}
\usetikzlibrary{arrows.meta,positioning,calc,decorations.pathreplacing,fit,backgrounds,shapes.geometric}
\usepackage{quantikz}
\usepackage{orcidlink}

\InputIfFileExists{numbers}{}{}

\providecommand{\numRho}{0.5549}
\providecommand{\numRhoCiLo}{0.552}
\providecommand{\numRhoCiHi}{0.558}
\providecommand{\numRhoData}{0.5756}
\providecommand{\numResidF}{0.584}
\providecommand{\numD}{0.895}
\providecommand{\numDCiLo}{0.830}
\providecommand{\numDCiHi}{0.952}
\providecommand{\numFfloor}{0.1310}
\providecommand{\numFfloorSd}{0.0037}
\providecommand{\numEmodel}{0.0137}
\providecommand{\numWoneCell}{0.0039}

\providecommand{\numWoneE}{0.0141}
\providecommand{\numWoneEpct}{0.9}
\providecommand{\numWoneEwidth}{10.9}
\providecommand{\numWidthData}{1.030}
\providecommand{\numWidthModel}{1.130}
\providecommand{\numWidthHw}{1.091}
\providecommand{\numPermP}{0.010}
\providecommand{\numNgen}{20\,000}
\providecommand{\numStripT}{0.0188}
\providecommand{\numStripFree}{0.662}
\providecommand{\numStripFreeMean}{0.112}
\providecommand{\numStripFloor}{0.0236}
\providecommand{\numTailQ}{0.05}
\providecommand{\numLamLdata}{0.231}
\providecommand{\numLamLdataCiLo}{0.208}
\providecommand{\numLamLdataCiHi}{0.250}
\providecommand{\numLamUdata}{0.134}
\providecommand{\numLamUdataCiLo}{0.120}
\providecommand{\numLamUdataCiHi}{0.146}
\providecommand{\numLambdaData}{0.097}
\providecommand{\numLamLmodel}{0.174}
\providecommand{\numLamUmodel}{0.114}
\providecommand{\numLambdaModel}{0.060}
\providecommand{\numLamLerr}{0.057}
\providecommand{\numLamUerr}{0.019}
\providecommand{\numEllipMaxDev}{0.024}
\providecommand{\numLamLabM}{0.27}
\providecommand{\numLamLabD}{0.61}
\providecommand{\numLamUabM}{0.33}
\providecommand{\numLamUabD}{0.61}
\providecommand{\numLamLbcM}{0.65}
\providecommand{\numLamLbcD}{0.70}
\providecommand{\numLamUbcM}{0.51}
\providecommand{\numLamUbcD}{0.44}
\providecommand{\numLamLghM}{0.49}
\providecommand{\numLamLghD}{0.70}
\providecommand{\numLamUghM}{0.54}
\providecommand{\numLamUghD}{0.56}
\providecommand{\numGresidF}{0.206}
\providecommand{\numGD}{0.982}
\providecommand{\numGlamLerr}{0.051}
\providecommand{\numGlamUerr}{0.045}
\providecommand{\numIresidF}{4.612}
\providecommand{\numID}{0}
\providecommand{\numIlamLerr}{0.181}
\providecommand{\numIlamUerr}{0.084}
\providecommand{\numHwD}{0.873}
\providecommand{\numHwDCiLo}{0.674}
\providecommand{\numHwDCiHi}{0.964}
\providecommand{\numHwRho}{0.5246}
\providecommand{\numHwRhoCiLo}{0.499}
\providecommand{\numHwRhoCiHi}{0.545}
\providecommand{\numHwFrob}{0.793}
\providecommand{\numHwGerms}{500}
\providecommand{\numHwShots}{1152}
\providecommand{\numHwPulses}{42}
\providecommand{\numHwLayers}{12}
\providecommand{\numHwSeffPred}{537}
\providecommand{\numHwSeffMeas}{1026}

\providecommand{\numHwChi}{2.42}
\providecommand{\numHwAtten}{0.986}
\providecommand{\numHwRhoDeconv}{0.5315}
\providecommand{\numHwRzzErr}{1.4\times10^{-3}}
\providecommand{\numGLambda}{0.001}

\providecommand{\numGrho}{0.5821}
\providecommand{\numHwLamL}{0.189}
\providecommand{\numHwLamU}{0.111}
\providecommand{\numHwLambda}{0.077}
\providecommand{\numHwGain}{0.944}
\providecommand{\numHwGainCiLo}{0.871}
\providecommand{\numHwGainCiHi}{0.993}
\providecommand{\numHwSeff}{\numHwSeffMeas}
\providecommand{\numHwResidF}{\numHwFrob}
\providecommand{\numBellSmax}{2.76}
\providecommand{\numBellFrac}{23\%}
\providecommand{\numBellShot}{2.793}
\providecommand{\numBellShotSe}{0.020}
\providecommand{\numBellSigma}{39}

\graphicspath{{figures/}{figs/}{./}}

\hypersetup{colorlinks=true, linkcolor=blue!60!black, citecolor=blue!60!black, urlcolor=blue!60!black}

\newcommand{\MMD}{\mathrm{MMD}}
\newcommand{\R}{\mathbb{R}}
\newcommand{\cO}{\mathcal{O}}
\newcommand{\cU}{\mathcal{U}}

\newcommand{\cA}{\mathcal{A}}

\newcommand{\Pind}{\Pi}                     
\newcommand{\RZZ}{R_{ZZ}}
\newcommand{\GEANT}{\textsc{Geant4}}
\newcommand{\E}{\mathbb{E}}
\theoremstyle{plain}
\newtheorem{proposition}{Proposition}
\newtheorem{corollary}{Corollary}

\begin{document}
	
	\title{Interferometric Quantum Polynomial Chaos Expansion as a\\
		Generative Model for Calorimeter Shower Simulation}
	
	\author{Jamal Slim\orcidlink{0000-0002-9418-8459}}
	\email{jamal.slim@desy.de}
	\affiliation{Deutsches Elektronen-Synchrotron DESY, 22603 Hamburg, Germany}

	\author{Saverio Monaco\orcidlink{0000-0001-8784-5011}}
	\affiliation{Deutsches Elektronen-Synchrotron DESY, 22603 Hamburg, Germany}
	\affiliation{RWTH Aachen University, 52062 Aachen, Germany}
	
	\author{Florian Rehm\orcidlink{0000-0002-8337-0239}}
	\affiliation{European Organization for Nuclear Research (CERN), 1211 Geneva, Switzerland}
	
	\author{Dirk Kr\"ucker\orcidlink{0000-0003-1610-8844}}
	\affiliation{Deutsches Elektronen-Synchrotron DESY, 22603 Hamburg, Germany}
	
	\author{Frank Gaede \orcidlink{0000-0002-7055-9200}}
	\affiliation{Deutsches Elektronen-Synchrotron DESY, 22603 Hamburg, Germany}
	
	\author{Kerstin Borras\orcidlink{0000-0003-1111-249X}}
	\affiliation{Deutsches Elektronen-Synchrotron DESY, 22603 Hamburg, Germany}
	\affiliation{RWTH Aachen University, 52062 Aachen, Germany}

	\begin{abstract}
		We present the quantum polynomial chaos expansion, a generative algorithm in which a single circuit is the entire model, and we use it to learn calorimeter images. In a classical chaos expansion the randomness is the input and the coefficients are fitted. Here the randomness is still the only input, entering the circuit as rotation angles and re-uploaded at every block, so that each measured observable is a chaos expansion of the latent variables whose order equals the circuit depth, and what is fitted are the gate angles themselves. Expressivity therefore grows with depth rather than with classical coefficients, correlations between outputs arise only from entangling gates, and a single latent wire read by all qubits carries the collective mode of the data. Nothing fitted stands between the circuit and the sample, so switching the entanglers off is a setting of the model itself and provably yields independent outputs, and attribution of the learned correlations to individual gates becomes a measurement. Choosing between two measurement bases shot by shot sharpens attribution into certification, and the trained model violates the Bell bound obeyed by every classical generative model with local response, whatever its size. We train the model on Geant4 shower data, execute the identical circuit on a superconducting processor with its accuracy loss predicted in advance, prove a no-go theorem for the tail dependence of every smooth generator read out through expectation values, and identify the circuit primitive that removes this limit.
	\end{abstract}
	
	\maketitle
	
	% ============================================================================
	\section{Introduction}
	\label{sec:intro}
	
	Calorimeter simulation is one of the heaviest computing loads in high-energy physics~\cite{GEANT4_2003,GEANT4_2006,ATLAS_sim_2010}, and the $\cO(10^{10})$ events per year expected at the High-Luminosity LHC~\cite{HLLHCComputing_2020} have made fast surrogates a practical necessity rather than a curiosity. Classical generative models meet the need. Adversarial networks~\cite{CaloGAN_2017,CaloGAN_2018}, variational autoencoders~\cite{CaloVAE_2018}, normalizing flows~\cite{CaloFlow_2021,CaloFlow_2022} and diffusion models~\cite{CaloDiffusion_2023,CaloScore_2023} all reproduce shower images at $10^4$ to $10^7$ trainable parameters. The appeal of a quantum generator is that a circuit of $n$ qubits carries a distribution over a space of dimension $2^n$ while holding only $\cO(n)$ angles, so the same task might be done with a parameter count smaller by orders of magnitude. Born machines~\cite{BornMachine_2018,BornMachine_2019}, quantum adversarial networks~\cite{QGAN_Zoufal_2019,QGAN_Lloyd_2018} and expectation-value samplers, which turn classical input randomness into continuous outputs through measured observables~\cite{Romero_Aspuru_2021,Barthe_2024}, all pursue this. None has yet delivered a practical benefit at realistic scale, and the reasons are increasingly well documented~\cite{SchuldKilloran_2022,Kubler_2021,Bowles_2024}.
	
	\subsection{The attribution problem}
	
	Among quantum models, calorimeter surrogates that reach the largest problem sizes are hybrid quantum-classical methods, using an annealer to sample a restricted Boltzmann machine and a trained convolutional decoder to reconstruct the image~\cite{CaloQVAE_2024,Calo4pQVAE_2025}. In the largest of them a latent of $2048$ binary units becomes a shower of $6480$ voxels through the decoder~\cite{Calo4pQVAE_2025}. Fully quantum generators, which read each sample directly from measurement, have so far been shown only on heavily downsampled targets, an eight-pixel image from twenty-three trainable angles in the most developed case~\cite{QAG_2024,DualPQC_2021}.
	
	In a hybrid model the classical network carries most of the fitted capacity, and when the composite reproduces a target distribution the agreement belongs to the composite. It says little about the circuit, because an expressive classical map reaches the same target from many different feature sets, quantum or otherwise, and for broad classes of re-uploading circuits one can construct an efficient classical surrogate that reproduces the quantum model's own input-output relation~\cite{Schreiber_2023,Huang_2021}. Systematic benchmarking supports the concern. Removing entanglement from quantum models often leaves their performance unchanged, and among papers claiming to outperform a classical method only $4\%$ report a negative result~\cite{Bowles_2024}. What is missing is not a better score, but a \emph{control}, that is a measurement of what the model achieves with its quantum resource switched off, against which a score can be read. A hybrid model cannot supply one, since switching off the circuit leaves a trained network that still fits the data.
	
	Fully quantum models are free of this difficulty by construction, and Born machines are the established example, since every sample they emit is a projective measurement outcome~\cite{BornMachine_2018,BornMachine_2019,Coyle_2020}. What they return, however, are bitstrings rather than calibrated intensities, which has kept them to low-dimensional or coarsely discretized targets. Calorimeter images have therefore been left with a choice between a hybrid model whose classical part absorbs the attribution and a fully quantum model whose readout cannot express the data. This paper removes that choice.
	
	\subsection{Our approach}
	
	We make the circuit the generator. It produces the per-cell densities and the correlations between cells in a single execution. Each of the $n$ calorimeter cells is assigned one qubit, and the intensity of a cell is read from the expectation value of the Pauli-$Z$ operator on its qubit, shifted and scaled by two stored constants that carry the dimensionless expectation into physical units. Those $2n$ constants are computed once from the training data, are never fitted, and are the only classical numbers the deployed model contains.
	
	The quantum polynomial chaos expansion, QPCE, runs as follows. Draw $n$ private uniform germs, which are the latent random variables of a chaos expansion, together with one shared germ. Execute a single interferometric circuit that re-uploads them at the head of each of its $L$ blocks, so that $L$ is at once the circuit depth and, as we prove, the order of the expansion the circuit realizes. Measure once, and convert the expectation values into intensities with the fixed affine map. At the end of every fit, run the same circuit again with its couplers switched off and the shared germ frozen at its median, and confirm that the generated cells become exactly independent. Run it a third time with the shared germ released, which isolates the common-factor channel the design deliberately provides. One germ vector is one shower, the expectation values form a degree-$L$ chaos expansion of the germs, and both controls are parameter settings of the generator rather than baselines imported from outside (Algorithm~\ref{alg:qpce}).
	
	Two ingredients carry the construction. The first is \emph{germ re-uploading}. Encoding the latent noise a single time caps every expectation value at trigonometric degree one in each germ, whatever the depth, which is far too coarse for a skewed calorimeter marginal. Re-uploading the germ once per block makes each expectation a trigonometric polynomial of exact degree $L$ in that germ (Sec.~\ref{sec:circuit}), which is a chaos expansion whose order is bought with circuit depth. The second is a loss with exact gradients. The energy distance, a kernel measure of the gap between the generated and target distributions that we define in Sec.~\ref{sec:loss}, is minimized by L-BFGS with adjoint gradients costing three state evolutions irrespective of the number of parameters.
	
	\subsection{Summary of results}
	
	The model is quantum in its marginals and in its dependence alike. Every density and every unit of correlation comes from one circuit, read through the fixed affine map from each qubit's Pauli-$Z$ expectation to its cell intensity. The count of fitted classical parameters is exactly zero, and the stored classical count is the $2n$ unit conversions (Sec.~\ref{sec:model}).
	
	Depth is the truncation order. With the germs re-uploaded at each block, the expectation $\langle Z_i\rangle$ of qubit $i$ is a trigonometric polynomial of degree $L$ in every germ, which we verify to machine precision. The model therefore inherits the order parameter of classical polynomial chaos, with exponentially many coefficients constrained to be coherent functions of $\cO(nL)$ angles (Sec.~\ref{sec:circuit}).
	
	Dependence flows through two channels placed there by design, the two-qubit couplers and one shared germ wire that every qubit reads, and the verification protocol separates them (Sec.~\ref{sec:verification}). Training uses a single loss and exact gradients. The energy distance equals twice the squared maximum mean discrepancy ($\MMD$) taken with the parameter-free distance kernel, which we confirm numerically to $1.9\times10^{-14}$, and the adjoint gradients agree with the parameter-shift rule~\cite{ParamShift_Mitarai_2018,ParamShift_Schuld_2019} to $10^{-14}$ at $28\times$ lower cost. Since the frozen-batch objective continues to fall after the model has stopped improving, we validate on fresh germs, stop early, and keep the best-validation iterate. Every number we report is computed on fresh germs against the held-out split (Sec.~\ref{sec:loss}).
	
	The model carries its own controls. Remove the couplers and freeze the shared wire, and the cells are exactly independent whatever the remaining angles (Proposition~\ref{prop:floor}). Measuring dependence throughout by the Spearman rank correlation $\rho_S$ between cell pairs, the deployed checkpoint in this configuration reaches $\max|\rho_S|=\numStripT$ against a sampling ceiling of $\numStripFloor$, while the full model reaches mean $|\rho_S|=\numRho$ and a dependence score $D=\numD$, a score constructed in Sec.~\ref{sec:metrics} that equals one when the generated rank structure matches the data and zero at independence. With the couplers removed but the shared wire released, the common-factor channel alone carries mean $|\rho_S|=\numStripFreeMean$. The attribution is thus a decomposition of the generated dependence into an entangler channel and a shared-latent channel, and we report both (Sec.~\ref{sec:verification}). On hardware the shot-noise contraction of each correlation is predicted analytically before deployment, the readout calibration is measured by dedicated circuits that never touch the model, and by Proposition~\ref{prop:contraction} no per-cell correction, however obtained, can change any dependence statistic (Sec.~\ref{sec:hardware}).
	
	The deployed circuit is designed for the machine that runs it. The natural dense alternative, a ring-plus-skip-$2$ coupler graph, embeds in no present device, and routing it inflates the two-qubit depth by a factor of $6$ over its own chromatic-index floor. We use instead a seven-edge path, native to the heavy-hex lattice with zero SWAPs and edge-coloured to that floor, giving $\numHwPulses$ fractional two-qubit phase pulses, the native $\RZZ$ gate, in $\numHwLayers$ two-qubit layers and about $1\,\mu$s of circuit time (Sec.~\ref{sec:topology}). Executed unchanged on \texttt{ibm\_fez}, it reaches $D=\numHwD$ $[\numHwDCiLo,\numHwDCiHi]$ and mean $|\rho_S|=\numHwRho$, against $\numRho$ in noiseless simulation (Sec.~\ref{sec:hardware}).
	
	We are equally precise about what the construction is. At $n=8$ the model is exactly simulable classically, with light-cone cost $\cO(n\,2^{4L+1})$ for this ansatz class, and we make no hardness claim in either direction (Sec.~\ref{sec:simulability}). Classical reference densities appear only as correctness checks that fix the scale of each metric (Sec.~\ref{sec:honest}). One structural property of the target deserves emphasis, its tail asymmetry $\lambda_L-\lambda_U=\numLambdaData$ at $q=\numTailQ$ with disjoint bootstrap intervals, a quantity every elliptical family fixes at zero identically. The circuit reproduces $\numLambdaModel$ of it at the same quantile, and Proposition~\ref{prop:tailnogo} characterizes the asymptotic tail behaviour of the entire smooth expectation-value class (Sec.~\ref{sec:tails}).
	
	% ============================================================================
	\section{From Classical Polynomial Chaos to a Fully Quantum Expansion}
	\label{sec:analogy}
	% ============================================================================
	
	\subsection{Classical PCE and its separability limitation}
	\label{sec:pce_background}
	
	Polynomial chaos expansion represents a square-integrable random variable as a spectral expansion in orthogonal polynomials of independent standardized germs~\cite{Wiener_1938,Xiu_Karniadakis_2002,Xiu_2010}. Each output is expanded as
	\begin{equation}
		Y_j \approx \sum_{\bm{\alpha}\in\cA} c_{j,\bm{\alpha}}\,\Psi_{\bm{\alpha}}(\bm{U}),
		\label{eq:pce_general}
	\end{equation}
	with a multivariate basis of product form, $\Psi_{\bm{\alpha}}(\bm{U})=\prod_j P_{\alpha_j}(U_j)$. That basis is \emph{separable}, so inter-variable dependence must be built explicitly from cross terms $U_iU_j$, $U_iU_jU_k$, $\dots$, and the dictionary grows as $\binom{n+p}{p}$ at degree $p$.
	
	\subsection{Two distinct quantum analogues}
	\label{sec:two_analogues}
	
	A quantum analogue of Eq.~\eqref{eq:pce_general} can be built in two ways, and they differ in where the fitted numbers live.
	
	\emph{Quantum basis, classical coefficients.} The polynomial basis $\Psi_{\bm{\alpha}}$ is replaced by Pauli expectations $\langle O_k\rangle(\bm{U})$ of a data-encoded state, while the coefficients $c_{j,\bm{\alpha}}$ remain a classical linear solve. This is the standard hybrid construction, in which the quantum device serves as a feature map and a trained classical network produces the output. It is also where the attribution problem described in Sec.~\ref{sec:intro} originates.
	
	\emph{Relation to prior work.} Expectation-value readout as a generative mechanism was proposed by Romero and Aspuru-Guzik~\cite{Romero_Aspuru_2021}, and Barthe \emph{et al.}~\cite{Barthe_2024} analyzed the resulting model class in generality, proving universality and resource bounds and establishing that re-uploading circuits with integer generator spectra emit variables with an exact finite chaos expansion. Quantum-circuit encodings of classical chaos surrogates for function approximation appear in Ref.~\cite{Aftab_Schwab_2025} with a complementary aim. Barthe and coworkers separate two readout mechanisms, Born-rule sampling, in which the randomness of a sample is the measurement itself, and expectation-value sampling, in which classical randomness enters the circuit and the output is a measured expectation~\cite{Barthe_2024}. QPCE belongs to the second class, and this is not incidental. A chaos expansion is by definition a map from prescribed input randomness to an output, so a quantum chaos expansion must take its randomness as input, which is exactly what makes depth able to act as the expansion order. The present construction differs from all of these works in where the randomness lives and in what is fitted. In a classical chaos expansion the germs are given and the coefficients are solved for. In a hybrid model the data are encoded into the circuit and a classical network is fitted after it. QPCE inverts both at once. Noise is the only input the circuit ever receives, injected on dedicated wires so that depth becomes the expansion order, the data touch the model only through the loss, and the fitted objects are the gate angles, so one circuit is simultaneously the basis, the coefficients and the sampler. The latent structure itself becomes an architectural resource. Private wires carry cell-local randomness and one shared wire, read by every qubit with trainable weight, realizes the common factor of the data, a design freedom that classical expansions do not possess because their germs carry no topology. Finally, where the prior analyses establish existence, this paper builds the object and holds it to account. Depth equals order with measured residuals, the verification null is a parameter setting of the trained generator, the hardware run carries an error budget predicted before submission, and a no-go theorem marks the boundary of the entire smooth expectation-readout class, proved against our own model family.
	
	\emph{QPCE, quantum basis \emph{and} quantum coefficients.} The expansion is inverted. Rather than expanding $Y$ in quantum features and fitting classical weights, the circuit generates the showers directly and the \emph{gate angles} are fitted to the data. There are no coefficients to solve for because there is no expansion to enumerate. The quantum state is the model, and the measured expectation values are already the observable. Table~\ref{tab:analogy} makes the correspondence explicit.
	
	\begin{table}[t]
		\centering
		\caption{Structural correspondence between classical PCE and QPCE. The entries in bold are where QPCE departs both from classical PCE and from the hybrid construction introduced in Sec.~\ref{sec:two_analogues}.}
		\label{tab:analogy}
		\renewcommand{\arraystretch}{1.3}
		\begin{tabular}{p{3.5cm} p{4.1cm}}
			\toprule
			\textbf{Classical PCE} & \textbf{QPCE} \\
			\midrule
			Data embedding of $\bm U$ into the basis & \textbf{eliminated}, the circuit receives \textbf{no data} \\
			Germ / input randomness & \textbf{in-circuit noise wires} $R_y(\pi\varepsilon_i)$, $\varepsilon\!\sim\!\cU(0,1)$ \\
			Polynomial basis $\Psi_{\bm{\alpha}}(\bm U)$ & \textbf{the state itself}, no basis is enumerated \\
			Legendre orthogonality & Fourier-cosine orthogonality on $(0,1)^n$ \\
			Separable product basis & \textbf{entangled state}, dependence from $\RZZ$ couplers and one shared germ wire \\
			Degree-$p$ truncation & interferometer depth $L$ \\
			Coefficients $c_{j,\bm{\alpha}}$ (fitted) & \textbf{gate angles $\bm\beta,\bm\phi,\bm\theta$ (fitted)} \\
			Ridge solve & \textbf{L-BFGS on the energy distance} \\
			Reconstruction $\sum c\Psi$ & \textbf{linear readout} $Y_i=m_i+s_i\langle Z_i\rangle$ \\
			\bottomrule
		\end{tabular}
	\end{table}
	
	The consequence for the separability limitation discussed in Sec.~\ref{sec:pce_background} is direct. Classical PCE reaches dependence by enumerating cross terms. QPCE reaches it by turning on a two-qubit phase and, for the global mode of the data, by feeding one extra germ to every qubit. The coupler channel carries $L|E|=42$ phases here and, crucially for what follows, setting them to zero and freezing the shared wire returns the model \emph{exactly} to independence rather than to a poorly fitted approximation of it. That exactness is what makes verification possible, and it is the reason we use continuous $\RZZ$ couplers rather than CNOTs (Sec.~\ref{sec:whyrzz}).
	
	\subsection{Generation without data embedding}
	\label{sec:scope}
	
	A hybrid construction encodes observations into rotation angles, which makes the circuit a data-dependent feature map whose output must then be decoded. In QPCE that embedding is absent. The circuit input is $\bm\varepsilon\sim\cU(0,1)^n$, i.i.d.\ noise drawn afresh for every generated sample, and nothing else. The state $\ket{\psi(\bm\varepsilon)}$ is conditioned on no observation. Training data touch the model at exactly one point, the loss $\mathcal E(\hat{\bm Y},\bm Y)$, which is evaluated after measurement and never inside a quantum register. With no data in the circuit, any structure in its output, marginal or cross-cell, must have been generated by gates, and in Sec.~\ref{sec:verification} we verify this for the dependence.
	
	% ============================================================================
	\section{The QPCE Model}
	\label{sec:model}
	% ============================================================================
	
	\subsection{The algorithm}
	\label{sec:algorithm}
	
	QPCE consists of three procedures, \textsc{Train}, \textsc{Generate} and \textsc{Verify}, sharing one circuit and one parameter vector. Verification is the third procedure rather than an appendix because it runs at the end of every fit. Algorithm~\ref{alg:qpce} states all three completely.
	
	\begin{table}[t]
		\algcaption{The complete QPCE algorithm. No step contains a fitted
			classical parameter. \textsc{Generate} is the entire deployed model,
			and \textsc{Verify} reuses the generation circuit at two parameter
			settings, so it runs on any backend the model runs on.}%
		\label{alg:qpce}
		\begin{algorithmic}[1]
			\Require training showers $\bm Y\in\R^{N\times n}$; blocks $L$;
			edge set $E$; shots $S$
			\Ensure gate angles $(\bm\beta,\bm\phi,\bm\theta,\bm a,\bm w)$;
			readout constants $(\bm m,\bm s)$
			\Statex
			\Procedure{Train}{$\bm Y$}\Comment{run once}
			\State $m_i\gets\tfrac12\bigl(\max_b Y_{bi}+\min_b Y_{bi}\bigr)$,\quad
			$s_i\gets\tfrac12\bigl(\max_b Y_{bi}-\min_b Y_{bi}\bigr)$
			\Statex \hfill\emph{stored once, never revisited}
			\State freeze a germ batch $\{(\bm\varepsilon_b,\varepsilon_{0,b})\}_{b=1}^{B}\!\sim\!\cU(0,1)^{n+1}$
			and a shower batch $\{\bm Y_b\}$
			\State $(\bm\beta,\bm\phi,\bm\theta,\bm a,\bm w)\gets\arg\min\;
			\mathcal{E}\bigl(\bm m+\bm s\odot\langle\bm Z\rangle,\;\bm Y\bigr)=2\,\MMD^2$
			\Statex \hfill\emph{L-BFGS, exact adjoint gradients; validated on
				fresh germs, early-stopped, best-validation iterate saved}
			\EndProcedure
			\Statex
			\Procedure{Generate}{\,}\Comment{one germ vector $\to$ one shower}
			\State draw $(\bm\varepsilon,\varepsilon_0)\sim\cU(0,1)^{n+1}$
			\State run the circuit of Eq.~\eqref{eq:circuit}, re-uploading
			$(\bm\varepsilon,\varepsilon_0)$ in each of the $L$ blocks
			\State measure all qubits; estimate $\langle Z_i\rangle$
			\Statex \hfill\emph{$S$ shots on hardware, exact in simulation}
			\State \Return $Y_i=m_i+s_i\langle Z_i\rangle$
			\Comment{degree-$L$ chaos expansion, Eq.~\eqref{eq:chaos}}
			\EndProcedure
			\Statex
			\Procedure{Verify}{\,}\Comment{automatic at the end of every fit}
			\State regenerate with $\phi_{\ell,ij}\gets0$ for all $\ell,(i,j)\in E$
			\emph{and} $\varepsilon_0\gets\tfrac12$ frozen; all other angles untouched
			\State \textbf{assert} $\max_{i\neq j}\lvert\mathrm{corr}(u_i,u_j)\rvert$
			lies at the sampling floor \Comment{Proposition~\ref{prop:floor}}
			\State regenerate with $\phi\gets0$, $\varepsilon_0$ \emph{free};
			record the shared-latent channel alone
			\State \Return the effect size $D$ of Eq.~\eqref{eq:effectsize}
			and both strip statistics
			\EndProcedure
		\end{algorithmic}
	\end{table}
	
	The three procedures share one circuit and one parameter vector. \textsc{Generate} is the deployed model in its entirety, no transform precedes the circuit and none follows it, so the parameter audit of the algorithm is in fact the parameter audit of the circuit. Its output is a chaos expansion of degree $L$ because step 2 re-uploads the germs in every block, which is the precise sense in which depth is the truncation order. \textsc{Verify} is not a separately trained baseline but the same circuit evaluated at $\bm\phi=0$, a setting that forces exact independence conditional on the shared germ by Proposition~\ref{prop:floor}. Freezing $\varepsilon_0$ realizes the conditioning operationally, and releasing it isolates the shared-latent channel. The null model is a limit point of the trained family rather than an external comparison. Figure~\ref{fig:architecture} draws the pipeline these procedures execute, and we report in Sec.~\ref{sec:strip} we report the measured decomposition they produce.

	\subsection{Pipeline}
	\label{sec:pipeline}
	
	Figure~\ref{fig:architecture} draws the algorithm as it executes. Training data $\bm{Y}\in\R^{N\times n}$ enter only through the loss. At generation time $n$ independent uniform germs $\bm\varepsilon$ are drawn, the circuit is executed, and the shower is read out linearly, $Y_i=m_i+s_i\langle Z_i\rangle$, with $(m_i,s_i)$ fixed unit-conversion constants set once from the training range and never fitted. No transform precedes the circuit, none follows it, and no training datum is resampled at generation. Generation is fully unconditional.

	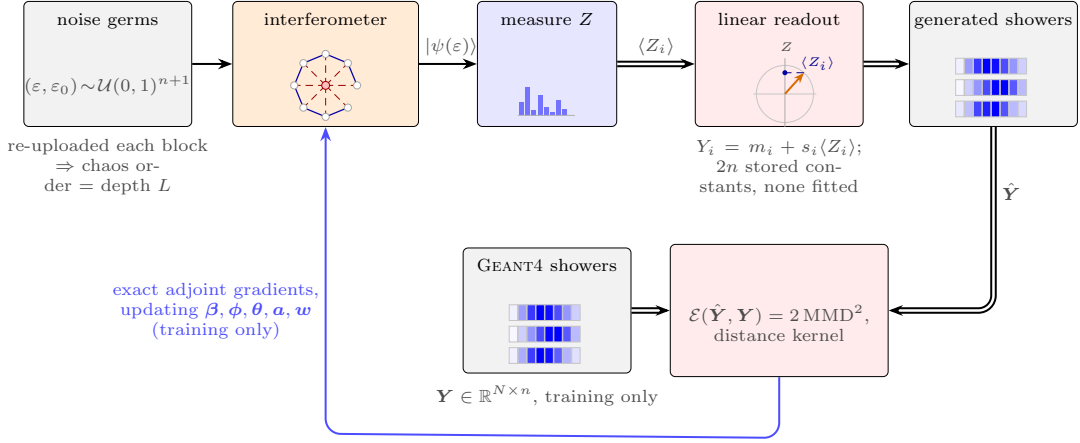
\begin{figure*}[t]
		\centering
		\resizebox{0.80\textwidth}{!}{%
			\begin{tikzpicture}[
				font=\small,
				q/.style={draw, rounded corners=2pt, fill=orange!16, minimum height=1.75cm,
					minimum width=2.60cm, font=\scriptsize, align=center},
				m/.style={draw, rounded corners=2pt, fill=blue!10, minimum height=1.75cm,
					minimum width=1.95cm, font=\scriptsize, align=center},
				r/.style={draw, rounded corners=2pt, fill=red!8, minimum height=1.75cm,
					minimum width=2.35cm, font=\scriptsize, align=center},
				d/.style={draw, rounded corners=2pt, fill=gray!10, minimum height=1.75cm,
					minimum width=2.35cm, font=\scriptsize, align=center},
				r2/.style={draw, rounded corners=2pt, fill=red!8, minimum height=1.85cm,
					minimum width=3.05cm, font=\scriptsize, align=center},
				qw/.style={-{Stealth[length=2mm]}, thick},
				cw/.style={-{Stealth[length=2mm]}, thick, double, double distance=1pt},
				fb/.style={-{Stealth[length=2mm]}, thick, blue!70},
				note/.style={font=\scriptsize, text=gray!60!black, align=center},
				lab/.style={font=\scriptsize, align=center},
				]
				\newcommand{\shower}[3]{%
					\begin{scope}[shift={(#1,#2)}]
						\foreach \v [count=\i from 0] in {#3}
						{\fill[blue!\v!white, draw=gray!45, line width=0.2pt]
							({\i*0.125},0) rectangle ++(0.118,0.22);}
				\end{scope}}
				\newcommand{\histicon}[4]{%
					\begin{scope}[shift={(#1,#2)}]
						\foreach \h [count=\i from 0] in {#4}
						{\fill[#3] ({\i*0.086},0) rectangle ++(0.068,{\h*0.40});}
						\draw[gray!60, line width=0.25pt] (-0.02,0) -- (0.78,0);
				\end{scope}}
				\newcommand{\blochicon}[3]{% x y radius
					\begin{scope}[shift={(#1,#2)}]
						\draw[gray!45,line width=0.3pt] (0,0) circle (#3);
						\draw[gray!65,line width=0.3pt] (0,-#3-0.06) -- (0,#3+0.10);
						\draw[gray!45,line width=0.25pt] (-#3-0.04,0) -- (#3+0.04,0);
						\node[font=\tiny,text=gray!60!black,anchor=south] at (0.02,{#3+0.08}) {$Z$};
						\draw[blue!55!black,densely dashed,line width=0.35pt]
						({sin(42)*#3},{cos(42)*#3}) -- (0,{cos(42)*#3});
						\fill[blue!65!black] (0,{cos(42)*#3}) circle (0.035);
						\node[font=\tiny,text=blue!55!black,anchor=west] at ({0.10},{cos(42)*#3+0.13}) {$\langle Z_i\rangle$};
						\draw[orange!85!black,line width=0.8pt,-{Stealth[length=1.5mm]}]
						(0,0) -- ({sin(42)*#3},{cos(42)*#3});
				\end{scope}}
				\newcommand{\ringiconoff}[3]{%
					\begin{scope}[shift={(#1,#2)}]
						\foreach \a in {0,45,...,315}{\fill[white,draw=gray!70,line width=0.3pt] (\a:#3) circle (0.05);}
				\end{scope}}
				\newcommand{\ringicon}[3]{%
					\begin{scope}[shift={(#1,#2)}]
						\foreach \a in {0,45,...,315}{\coordinate (r\a) at (\a:#3);}
						% seven-edge path drawn on the ring layout (ring left open: no 315--0 edge)
						\foreach \a/\b in {0/45,45/90,90/135,135/180,180/225,225/270,270/315}
						{\draw[blue!65!black, line width=0.5pt] (r\a) -- (r\b);}
						% shared germ wire: dashed chords from the central hub to every qubit
						\foreach \a in {0,45,...,315}
						{\draw[red!60!black, line width=0.25pt, densely dashed] (0,0) -- (r\a);}
						\fill[red!25,draw=red!60!black,line width=0.3pt] (0,0) circle (0.055);
						\node[font=\tiny,text=red!60!black,anchor=south west] at (0.05,0.5) {};
						\foreach \a in {0,45,...,315}{\fill[white,draw=gray!70,line width=0.3pt] (r\a) circle (0.05);}
				\end{scope}}
				
				% ---------------- row 1 (the deployed model, left to right) ----------------
				\node[d] (eps)  at (0,0)      {};
				\node[q] (circ) at (3.05,0)   {};
				\node[m] (meas) at (6.15,0)   {};
				\node[r] (ro)   at (9.40,0)   {};
				\node[d] (gen)  at (12.40,0)  {};
				
				\node[lab,anchor=north] at ($(eps.north)+(0,-0.05)$)  {noise germs};
				\node[lab,anchor=north] at ($(circ.north)+(0,-0.05)$) {interferometer};
				\node[lab,anchor=north] at ($(meas.north)+(0,-0.05)$) {measure $Z$};
				\node[lab,anchor=north] at ($(ro.north)+(0,-0.05)$)   {linear readout};
				\node[lab,anchor=north] at ($(gen.north)+(0,-0.05)$)  {generated showers};
				
				\node[note] at (0,-0.30) {$(\varepsilon,\varepsilon_0)\!\sim\!\cU(0,1)^{n+1}$};
				\node[note,anchor=north,text width=2.90cm] at ($(eps.south)+(0,-0.05)$)
				{re-uploaded each block\\$\Rightarrow$ chaos order $=$ depth $L$};
				\ringicon{3.05}{-0.30}{0.44}
				\histicon{5.76}{-0.72}{blue!55}{0.45,1.0,0.18,0.72,0.30,0.10,0.55,0.22}
				\blochicon{9.48}{-0.42}{0.40}
				\node[note,anchor=north,text width=2.90cm] at ($(ro.south)+(0,-0.05)$)
				{$Y_i=m_i+s_i\langle Z_i\rangle$;\\$2n$ stored constants, none fitted};
				\shower{11.87}{-0.14}{6,38,78,100,94,70,41,19}
				\shower{11.87}{-0.44}{3,19,41,71,98,100,65,27}
				\shower{11.87}{-0.74}{9,25,73,88,100,68,27,15}
				
				\draw[qw] (eps)  -- (circ);
				\draw[qw] (circ) -- node[above,note,pos=0.56]{$\ket{\psi(\varepsilon)}$} (meas);
				\draw[cw] (meas) -- node[above,note]{$\langle Z_i\rangle$} (ro);
				\draw[cw] (ro)   -- (gen);

				% ---------------- row 2 (training only) ----------------
				\node[d]  (data) at (6.15,-3.45) {};
				\node[r2] (loss) at (9.40,-3.45) {};
				\node[lab,anchor=north] at ($(data.north)+(0,-0.05)$) {\GEANT\ showers};
				\node[lab,anchor=north] at ($(loss.north)+(0,-0.05)$) {};
				\node[note] at (9.40,-3.62) {$\mathcal{E}(\hat{\bm Y},\bm Y)=2\,\MMD^2$,\\distance kernel};
				
				\shower{5.62}{-3.59}{6,38,78,100,94,70,41,19}
				\shower{5.62}{-3.89}{3,19,41,71,98,100,65,27}
				\shower{5.62}{-4.19}{6,33,76,100,94,63,29,12}
				\node[note,anchor=north] at ($(data.south)+(0,-0.05)$) {$\bm Y\in\R^{N\times n}$, training only};
				
				\draw[cw] (data) -- (loss);
				\draw[cw, rounded corners=6pt] (gen.south) -- (12.40,-3.45)
				node[pos=0.35, right, note]{$\hat{\bm Y}$} -- (loss.east);
				
				\draw[fb, rounded corners=6pt]
				(loss.south) -- ++(0,-0.85) -- (3.05,-5.175) -- (circ.south);
				\node[note,anchor=north west,text=blue!70,text width=3.60cm] at (-0.4,-2.95)
				{exact adjoint gradients,\\updating $\bm\beta,\bm\phi,\bm\theta,\bm a,\bm w$\\(training only)};
			\end{tikzpicture}%
		}
		\caption{QPCE drawn as it executes. The upper row is the entire deployed model. Germs enter the interferometer, one measurement setting yields $\langle Z_i\rangle$, and the fixed affine readout returns a shower. Grey marks data, orange the quantum circuit, blue the measurement and red the unit conversion. The interferometer icon is the actual coupler graph drawn on a ring layout, in which the solid edges are the seven-edge path (the ring is left open, $|E|=7$, native to the heavy-hex lattice with zero SWAPs), each edge carrying one trainable phase per block, and the dashed chords are the one shared germ wire $\varepsilon_0$, read by every qubit with a trainable amplitude at zero two-qubit cost. Re-uploading the germs at the head of every block makes the depth $L$ the chaos order (Eq.~\eqref{eq:chaos}). The lower row runs during training only, where \GEANT\ data reach the model through a single object, the parameter-free energy distance $\mathcal{E}=2\,\MMD^2$, whose exact adjoint gradient (blue) updates the $154$ gate angles. With the couplers off and $\varepsilon_0$ frozen, the graph reduces to bare vertices and the cells are exactly independent for any remaining angles (Proposition~\ref{prop:floor}), with the couplers off and $\varepsilon_0$ free, the fan alone remains, which is the shared-latent channel that \textsc{Verify} reports separately.}
		\label{fig:architecture}
	\end{figure*}
	
	\subsection{The interferometric ansatz}
	\label{sec:circuit}
	
	For $n$ qubits and $L$ blocks on an edge set $E$, the circuit is
	\begin{equation}
		\ket{\psi(\bm\varepsilon,\varepsilon_0)} \;=\; \prod_{\ell=1}^{L} W_\ell\, D_\ell\, G(\bm\varepsilon,\varepsilon_0)\,\ket{0^n},
		\label{eq:circuit}
	\end{equation}
	\begin{equation}
		G(\bm\varepsilon,\varepsilon_0)=\bigotimes_i R_y\!\bigl(\pi(w_i\varepsilon_i + a_i\varepsilon_0)\bigr),
		\label{eq:germ}
	\end{equation}
	\begin{equation}
		W_\ell=\bigotimes_i R_y(\theta_{\ell,i}),
		\qquad
		D_\ell=\!\!\prod_{(i,j)\in E}\!\!\RZZ(\phi_{\ell,ij})\prod_i R_z(\beta_{\ell,i}),
		\label{eq:blocks}
	\end{equation}
	read right to left within each block and followed by a single computational-basis measurement (Fig.~\ref{fig:circuit}). Here $\bm\varepsilon\in(0,1)^n$ are the i.i.d.\ private germs and $\varepsilon_0\in(0,1)$ is one additional i.i.d.\ germ \emph{shared by every qubit}, all \emph{re-uploaded by $G$ at the start of every block}. The block $D_\ell$ is diagonal, $\RZZ(\phi)=\exp(-i\phi Z_iZ_j/2)$, and $\{\bm\beta,\bm\phi,\bm\theta,\bm a,\bm w\}$ are the trainable angles ($154$ for the deployed instance, $n=8$, $L=6$, $|E|=7$). Every fitted number, including the germ amplitudes $a_i$ and $w_i$, is a rotation-angle coefficient on the circuit, so the classical trainable count remains zero. All $\RZZ$ in a block commute, so their ordering is immaterial, a fact exploited by the edge-coloured schedule constructed in Sec.~\ref{sec:topology}.
	
	The ansatz departs from the CNOT-plus-rotations layer that a hardware-efficient template would supply by default, and we motivate each departure in turn.

	\subsubsection{Germ re-uploading sets the chaos order}
	\label{sec:reupload}
	The decisive structural element is that the germ rotation opens \emph{every} block rather than only the first. By the Fourier structure of variational circuits~\cite{Schuld_2021_fourier}, an observable of a circuit that encodes a germ through $r$ rotations is a trigonometric polynomial with at most $r$ harmonics in that germ. Encoded once, every $\langle Z_i\rangle$ is degree one at \emph{any} depth (measured, fitting a single harmonic leaves relative residual $8\times10^{-16}$ at $L=2$ and $4\times10^{-15}$ at $L=6$), and a degree-one chaos expansion cannot represent a skewed calorimeter marginal. Re-uploaded $L$ times,
	\begin{equation}
		\langle Z_i\rangle(\bm\varepsilon,\varepsilon_0)\;=\;\sum_{\max_j|\omega_j|\le L} c_{\bm\omega}\,e^{i\pi\bigl(\sum_j \omega_j w_j\varepsilon_j\,+\,\omega_0 a\varepsilon_0\bigr)},
		\label{eq:chaos}
	\end{equation}
	a trigonometric chaos expansion of \emph{exact} degree $L$ in each germ wire, with the coefficients $c_{\bm\omega}(\bm\beta,\bm\phi,\bm\theta,\bm a,\bm w)$ coherent functions of the angles and the base frequencies $\pi w_j$, $\pi a_j$ set by the trainable amplitudes (measured on the unit-amplitude family, degree-$L$ fit residual $\le 2\times10^{-15}$ and degree-$(L{-}1)$ residual $\cO(1)$ for $L=2,3,4$). The Fourier-cosine system is complete and orthogonal on $(0,1)^{n+1}$ and, at unit amplitudes, spans the same space as tensor Legendre polynomials of the same degree, so $L$ plays exactly the role of the PCE order $p$, with two differences from classical practice. The truncation is per-coordinate rather than total-degree, and the exponentially many coefficients $c_{\bm\omega}$ are not free but coherent functions of $\cO(nL)$ angles. That constraint is the quantum content of the construction.
	
	\subsubsection{Noise amplitudes}
	Using a trainable encoding scale is \textit{risky} under any dependence-sensitive objective. For instance, shrinking a private noise amplitude increases output correlation monotonically, so a correlation-only loss would drive it to zero and produce near-perfect correlations by collapsing the generative diversity. In the private-only family we therefore fix the encoding scale to unity. With a shared wire the situation changes, the quantity the data pin is the shared \emph{variance fraction} of each cell (Sec.~\ref{sec:sharedgerm}), and that fraction is exactly the ratio that the amplitudes $w_i$ and $a_i$ control, so both are trainable. The collapse attractor is neutralized not by a projection but by the objective itself so that the loss is the full energy distance rather than a correlation score, and latent collapse widens the marginal mismatch and raises $\mathcal E$ directly. Measured on the converged deployed model no collapse occurs, the per-cell $W_1$ is $\numWoneCell$, and the fitted private amplitudes remain $\cO(0.3)$ to $\cO(0.6)$ on seven of the eight cells.
	
	\subsubsection{Continuous couplers rather than CNOTs}
	\label{sec:whyrzz}
	A CNOT is maximally entangling and has no dial, so a model family built on it contains no controlled null. $\RZZ(\phi)$ is a continuous entangler with $\phi=0$ giving \emph{exact} separability, conditional on the shared germ, and correlation growing smoothly with $|\phi|$. This is not merely convenient for optimization, it is actually what makes Proposition~\ref{prop:floor} an exact statement rather than an approximation, and hence what makes the verification protocol possible at all.
	
	\subsubsection{Circuit class and simulability}
	\label{sec:simulability}
	The circuit class of the ansatz determines which complexity results do and do not pertain to it, so we state it exactly and verify it directly. An IQP circuit is a single block of gates diagonal in one basis conjugated by fixed mixing layers, $U=H^{\otimes n}D\,H^{\otimes n}$, equivalently a circuit whose gates all commute. The ansatz of Eq.~\eqref{eq:circuit} is not of this form. The wall $R_y(\theta_{\ell,i})$ of block $\ell$ acts immediately before the germ rotation of block $\ell{+}1$ on the same qubit and the two compose on a single axis, $R_y(\theta)\,R_y(\pi\varepsilon)=R_y(\theta+\pi\varepsilon)$, so the circuit is a sequence of diagonal blocks separated by \emph{germ-dependent} $y$-rotation layers, a data-re-uploading brickwork. For the block structure of Eq.~\eqref{eq:circuit} at random angle settings, the conjugated operator $H^{\otimes n}U H^{\otimes n}$ carries off-diagonal weight of order one ($0.90$ for the deployed form at $L=2$, and $0.40$ even with the walls pinned at $\pi/2$, and $0.000$ for a true IQP control), and the germ and coupler layers do not commute.
	
	The simulability of the model is likewise stated exactly. The observable $\langle Z_i\rangle$ depends only on the causal cone of qubit $i$, which on a bounded-degree coupler graph grows linearly with depth, giving a light-cone simulation cost $\cO(n\,2^{4L+1})$, i.e., polynomial in $n$ at fixed $L$, exponential in $L$. At the operating point $n=8$, $L=6$ the cone covers the register and the model lives in a $2^8$-dimensional statevector. This simulability is a feature of the construction as a generative model, exact training gradients, exact verification, and byte-level agreement between the deployed circuit and its reference are all consequences of it, and every resource of the construction scales linearly, $(2n+|E|)L+2n=\cO(nL)$ angles on $\cO(n)$ couplers, one measurement setting regardless of $n$, germs entering as rotation angles with no amplitude encoding or state-preparation oracle, and a training gradient of three state evolutions independent of the parameter count, so the protocol exports unchanged to larger systems.

	\subsubsection{One shared germ wire}
	\label{sec:sharedgerm}
	The dataset's rank-correlation matrix has one eigenvalue carrying $64.1\%$ of the dependence, with eigenvector $(-0.33,-0.40,-0.42,-0.38,-0.16,+0.28,+0.40,+0.38)$. This is a global factor with a front/back sign flip, physically the shower-start depth. A rank-one reconstruction from that mode alone reproduces all $28$ off-diagonal correlations at correlation $0.970$. Cells $0$ and $7$ anti-correlate at $-0.48$ because they share a common shock, not because they are adjacent, and a model in which every latent is private to one qubit must transport that dependence through the coupler graph, where it decays with graph distance. On a sparse hardware-native graph this is the binding constraint, trained private-only, the path model's long-range correlations collapse toward zero.
	
	The shared wire supplies the global mode directly. One extra i.i.d.\ germ $\varepsilon_0$ enters the rotation angle of \emph{every} qubit in every block, Eq.~\eqref{eq:germ}, with a trainable per-qubit amplitude $a_i$. The private amplitude $w_i$ lets the optimizer set each cell's shared-variance fraction by shrinking the private term rather than over-rotating the shared one. This is not classical structure: $a_i$ and $w_i$ are rotation angles on the quantum circuit, in the same class as $\beta$, $\phi$ and $\theta$, and $\varepsilon_0$ is a ninth independent $\cU(0,1)$ noise wire, re-uploaded to every qubit where the ansatz already re-uploads eight such wires to one qubit each. It adds $2n$ single-qubit rotation coefficients and \emph{zero} two-qubit gates, zero SWAPs and zero two-qubit depth, so the hardware cost of the deployed circuit is identical with the wire on or off.

	With $\bm\phi=0$ the cells are independent only \emph{conditionally} on $\varepsilon_0$, so Proposition~\ref{prop:floor} and the strip test are stated conditionally, and the attribution of dependence to entanglement becomes a two-channel decomposition rather than a single claim. The price is stated because the alternative that keeps the attribution unconditional, an ancilla entangled with the register and traced out, has no cheap realization on a heavy-hex device, where every eight-qubit region is a forest and reaching all data qubits from one ancilla costs on the order of a hundred two-qubit layers.
	
	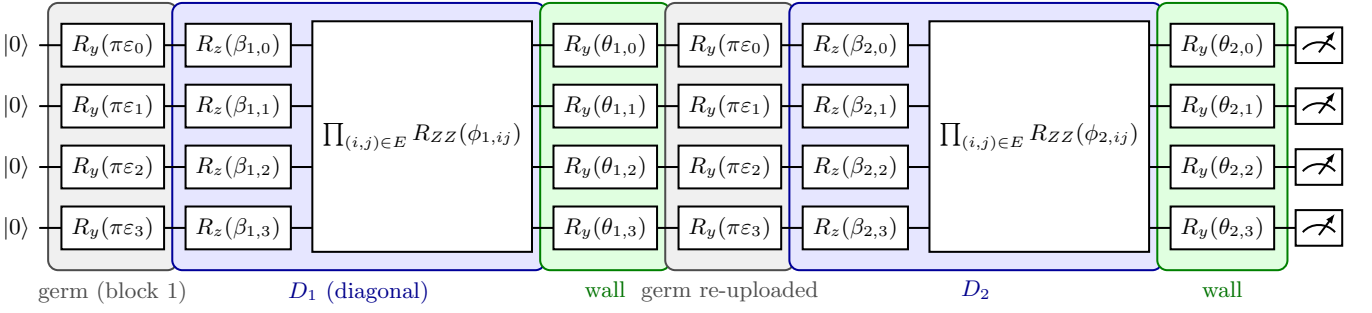
\begin{figure*}[tb]
		\centering
		\resizebox{\textwidth}{!}{%
			\begin{quantikz}[row sep={0.85cm,between origins}, column sep=0.30cm]
				\lstick{$\ket{0}$} & \gate{R_y(\pi\varepsilon_0)}\gategroup[4,steps=1,style={fill=gray!12,draw=gray!60!black,rounded corners,inner xsep=2pt,inner ysep=4pt},background,label style={label position=below,anchor=north,yshift=-0.22cm,text=gray!60!black}]{germ (block 1)}
				& \gate{R_z(\beta_{1,0})}\gategroup[4,steps=2,style={fill=blue!10,draw=blue!60!black,rounded corners,inner xsep=2pt,inner ysep=4pt},background,label style={label position=below,anchor=north,yshift=-0.22cm,text=blue!55!black}]{$D_1$ (diagonal)}
				& \gate[4]{\prod_{(i,j)\in E}\RZZ(\phi_{1,ij})}
				& \gate{R_y(\theta_{1,0})}\gategroup[4,steps=1,style={fill=green!12,draw=green!50!black,rounded corners,inner xsep=2pt,inner ysep=4pt},background,label style={label position=below,anchor=north,yshift=-0.22cm,text=green!45!black}]{wall}
				& \gate{R_y(\pi\varepsilon_0)}\gategroup[4,steps=1,style={fill=gray!12,draw=gray!60!black,rounded corners,inner xsep=2pt,inner ysep=4pt},background,label style={label position=below,anchor=north,yshift=-0.22cm,text=gray!60!black}]{germ re-uploaded}
				& \gate{R_z(\beta_{2,0})}\gategroup[4,steps=2,style={fill=blue!10,draw=blue!60!black,rounded corners,inner xsep=2pt,inner ysep=4pt},background,label style={label position=below,anchor=north,yshift=-0.22cm,text=blue!55!black}]{$D_2$}
				& \gate[4]{\prod_{(i,j)\in E}\RZZ(\phi_{2,ij})}
				& \gate{R_y(\theta_{2,0})}\gategroup[4,steps=1,style={fill=green!12,draw=green!50!black,rounded corners,inner xsep=2pt,inner ysep=4pt},background,label style={label position=below,anchor=north,yshift=-0.22cm,text=green!45!black}]{wall}
				& \meter{} \\
				\lstick{$\ket{0}$} & \gate{R_y(\pi\varepsilon_1)} & \gate{R_z(\beta_{1,1})} & & \gate{R_y(\theta_{1,1})} & \gate{R_y(\pi\varepsilon_1)} & \gate{R_z(\beta_{2,1})} & & \gate{R_y(\theta_{2,1})} & \meter{} \\
				\lstick{$\ket{0}$} & \gate{R_y(\pi\varepsilon_2)} & \gate{R_z(\beta_{1,2})} & & \gate{R_y(\theta_{1,2})} & \gate{R_y(\pi\varepsilon_2)} & \gate{R_z(\beta_{2,2})} & & \gate{R_y(\theta_{2,2})} & \meter{} \\
				\lstick{$\ket{0}$} & \gate{R_y(\pi\varepsilon_3)} & \gate{R_z(\beta_{1,3})} & & \gate{R_y(\theta_{1,3})} & \gate{R_y(\pi\varepsilon_3)} & \gate{R_z(\beta_{2,3})} & & \gate{R_y(\theta_{2,3})} & \meter{}
		\end{quantikz}}
		\vspace{0.35cm}
		\caption{QPCE interferometer, drawn for $n=4$, $L=2$ with unit germ amplitudes for legibility. In the deployed model each germ gate is $R_y(\pi(w_i\varepsilon_i+a_i\varepsilon_0))$, Eq.~\eqref{eq:germ}. The germ encodings open \emph{every} block, each germ wire therefore enters $L$ rotation angles, and $\langle Z_i\rangle$ is a trigonometric chaos expansion of exact degree $L$ in each wire (Eq.~\eqref{eq:chaos}). Trainable $R_z(\beta)$ and $\RZZ(\phi)$ form the diagonal blocks $D_\ell$. The walls $R_y(\theta_{\ell,i})$ are trainable. All $\RZZ$ within a block commute, so the coupler block is one multi-wire diagonal gate that schedules in $\chi'(G)$ parallel layers (Sec.~\ref{sec:topology}). Setting all $\phi_{\ell,ij}=0$ and freezing $\varepsilon_0$ makes the circuit a tensor product of single-qubit unitaries, each driven only by its own germ, which is the content of Proposition~\ref{prop:floor}. The deployed $n=8$ instance has $L=6$, $|E|=7$ on a path, $154$ angles.}
		\label{fig:circuit}
	\end{figure*}
	
	\subsection{Hardware-native entangler topology}
	\label{sec:topology}
	
	The natural first choice of coupler graph is the dense ring-plus-skip-$2$ graph $C_8(1,2)$, $E=\{(i,i\!+\!1)\}\cup\{(i,i\!+\!2)\}$ modulo $n$, $|E|=2n=16$, chosen as a geometric prior for lateral shower correlations. Confronted with hardware, that choice fails in a way worth reporting exactly, because it is generic to dense logical graphs. An exhaustive subgraph-embedding search shows $C_8(1,2)$ embeds in \emph{no} present device topology, zero embeddings on the heavy-hex lattice and zero on the square lattice, so every one of its couplers beyond a spanning structure must be SWAP-routed. Gate-level compilation of the routed circuit gives $316$ two-qubit pulses in $98$ two-qubit layers at $L=4$, against a commuting-block depth floor of $\chi'(C_8(1,2))\cdot L=16$. Routing inflates the two-qubit \emph{depth} by a factor of $6$, because a SWAP serializes a layer that was fully parallel. At a representative per-gate error the resulting circuit fidelity is $0.12$, and the shot budget for the full certification exceeds $10^8$. The dense graph is not an optimization target but is simply undeployable.
	
	The deployed model therefore inverts the design order. The coupler graph is chosen from the graphs the target device already has. All $\RZZ$ within a block commute, so a proper edge colouring of $E$ is a valid schedule with each colour class executing in parallel, and the two-qubit depth floor per block is the chromatic index $\chi'(G)$. On the heavy-hex lattice~\cite{IBMHeavyHex_2021}, whose girth is $12$, every eight-vertex subgraph is a forest and the eight-vertex path ($|E|=7$, $\chi'=2$) is the unique zero-SWAP choice. On a square lattice the ring ($\chi'=2$) and the $2\times4$ grid ($\chi'=3$) also embed. The deployed instance is the path at $L=6$ with $L|E|=42$ two-qubit pulses in $L\chi'=12$ layers, zero SWAPs, about $1.0\,\mu$s of circuit time, verified gate-by-gate against the exact statevector to $10^{-15}$ with the depth identity $\mathrm{depth}_{2q}=L\chi'$ asserted at build time. On a path the causal cone of $\langle Z_i\rangle$ is $\pm L$ sites, so depth buys range at two layers per unit of $L$ where the routed dense graph paid $24.5$. The long-range dependence the path cannot carry is supplied by the shared germ wire ( see Sec.~\ref{sec:sharedgerm}) at zero two-qubit cost.
	
	At deployment the logical path is placed by exhaustive search rather than by a heuristic. Specifically on \texttt{ibm\_fez}, the path has $136$ zero-SWAP embeddings, each scored by the sum of its calibrated two-qubit and readout error rates, errors add in the exponent of the fidelity product, so the sum is the correct objective, and the best-scoring eight-qubit region is selected. A data-driven alternative to the geometric prior, the maximum spanning tree of $|\tau_{ij}|$~\cite{Bedford_Cooke_2002,Joe_2014}, identifies one non-local coupler, cells $2$ and $6$ at $|\tau|=0.61$, that no hardware-native graph at $n=8$ can realize without routing. In the deployed model that dependence travels over the shared wire instead, which is the concrete sense in which the wire and the graph are complementary channels.
	
	\begin{figure}[t]
		\centering
		\begin{tikzpicture}[scale=0.78]
			\node[font=\small\bfseries] at (0,2.55) {Coupler graph + shared wire};
			\foreach \i in {0,...,7} {
				\node[circle, draw, fill=green!15, minimum size=0.60cm, font=\scriptsize\bfseries, inner sep=0] (q\i) at ({(\i-3.5)*1.05},0) {$q_\i$};
			}
			\foreach \i/\j in {0/1, 1/2, 2/3, 3/4, 4/5, 5/6, 6/7} {
				\draw[thick, blue!65!black] (q\i) -- (q\j);
			}
			\node[circle, draw, fill=red!12, minimum size=0.60cm, font=\scriptsize\bfseries, inner sep=0] (e0) at (0,1.65) {$\varepsilon_0$};
			\foreach \i in {0,...,7} {
				\draw[thin, red!60!black, densely dashed] (e0) -- (q\i);
			}
			\node[font=\small\itshape, text=blue!65!black] at (-2.6,-0.82) {path (7 edges, $\chi'=2$)};
			\node[font=\small\itshape, text=red!60!black] at (2.6,-0.82) {shared germ (0 two-qubit gates)};
			
			\node[font=\small\bfseries] at (0,-1.55) {Calorimeter cells};
			\foreach \i/\x in {0/-3.5, 1/-2.5, 2/-1.5, 3/-0.5, 4/0.5, 5/1.5, 6/2.5, 7/3.5} {
				\node[rectangle, draw, fill=orange!15, minimum size=0.72cm, font=\scriptsize\bfseries, inner sep=0] (p\i) at (\x,-2.45) {$p_\i$};
			}
			\foreach \i/\j in {0/1, 1/2, 2/3, 3/4, 4/5, 5/6, 6/7} {
				\draw[thin, orange!60!black, dashed] (p\i) -- (p\j);
			}
			\node[font=\small\itshape, text=orange!60!black] at (0,-3.50) {longitudinal shower profile, $q_i \leftrightarrow p_i$};
		\end{tikzpicture}
		\caption{Deployed coupler graph and the one-to-one map from qubits to cells. Solid lines mark the eight-qubit path, $|E|=7$, native to the heavy-hex lattice with zero routing SWAPs. Each edge carries one trainable phase per block, $L|E|=42$ coupler phases in total, and the commuting block schedules in $\chi'=2$ parallel layers. The dashed line marks the shared germ wire $\varepsilon_0$, read by every qubit with trainable amplitude $a_i$ at zero two-qubit cost, carrying the rank-one global mode of the data (Sec.~\ref{sec:sharedgerm}). Removing the couplers and freezing $\varepsilon_0$ makes the generated cells \emph{exactly} independent (Proposition~\ref{prop:floor}). Removing the couplers with $\varepsilon_0$ free isolates the shared-latent channel. The dense $C_8(1,2)$ alternative embeds in no present device and appears only as the routed reference we quote in Sec.~\ref{sec:topology}.}
		\label{fig:topology}
	\end{figure}
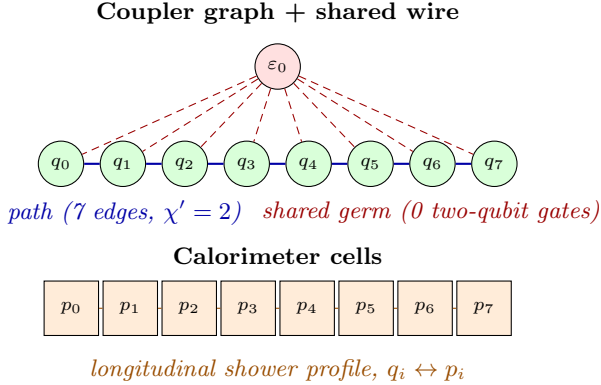
	
	\subsection{Single-setting linear readout}
	\label{sec:readout}
	
	QPCE measures the $n$ observables $\{Z_i\}$ from one computational-basis setting and converts them to intensities by the affine map
	\begin{equation}
		Y_i \;=\; m_i + s_i\,\langle Z_i\rangle ,
		\label{eq:readout}
	\end{equation}
	where $(m_i,s_i)$ are set once from the training range, midpoint and half-span per cell, and never fitted. They are unit-conversion constants of the same standing as a calorimeter's GeV-per-count calibration, $2n$ stored classical numbers and zero trainable ones. Because the intensities are read out from expectation values, the shot-noise cost of estimating them on hardware admits an exact and useful statement, the predicted shot contraction of every correlation (Sec.~\ref{sec:hardware}).
	
	\subsection{Training}
	\label{sec:loss}
	
	The model is trained by minimizing a single scalar, the energy distance~\cite{Szekely_2013,Baringhaus_2004}
	\begin{equation}
		\mathcal{E}(P,Q) = 2\,\E\lVert X-Y\rVert - \E\lVert X-X^\prime\rVert - \E\lVert Y-Y^\prime\rVert ,
		\label{eq:energy}
	\end{equation}
	between generated and Geant4 showers on per-cell standardized intensities. $\mathcal{E}$ is a metric on distributions. It vanishes iff $P=Q$, so marginals, correlations and all higher structure are consequences of one objective rather than separately weighted terms. The underlying kernel is the distance kernel $k(x,y)=\lVert x\rVert+\lVert y\rVert-\lVert x-y\rVert$, which is characteristic and carries no scale parameter, so no bandwidth selection enters the objective. This is a deliberate and standard choice rather than an omission. The energy distance is an established two-sample statistic~\cite{Szekely_2013,Baringhaus_2004}, and it equals $2\,\MMD^2$ with exactly this kernel~\cite{Sejdinovic_2013} (verified numerically to $1.9\times10^{-14}$), so the choice against the conventional Gaussian-kernel MMD~\cite{Gretton_2012_MMD,BornMachine_2018} is a choice of \emph{kernel} inside one framework, trading the Gaussian's tunable bandwidth for a kernel with none. Both are implemented. Trained to convergence the Gaussian mixture at $\{ \tfrac14,\tfrac12,1,2,4\}\times\sigma_{\mathrm{med}}$ gives a modestly better dependence residual at slightly worse marginals in the dense-graph reference configuration, and both kernels are implemented and reported.
	
	Under any joint loss the dependence contribution is small. On this dataset the loss gap between correlated and column-shuffled samples with identical marginals is $\mathcal{E}=0.067$, two orders below the initial marginal mismatch, so correlations move only near convergence. A capacity probe minimizing $\lVert\mathrm{Corr}(Z)-\mathrm{Corr}(Y)\rVert_F^2$ alone shows the dense encoding at $L=2$ already reaches mean $|\rho_S|=0.539$ of the data's $0.576$, so the encoding, not the optimizer, sets the ceiling we measure in Sec.~\ref{sec:densegraph}. 
	%	Any fixed iteration budget that terminates inside the marginal-fitting phase reads, falsely, as an expressivity failure of the ansatz, which is why exact gradients and full convergence are part of the algorithm rather than an implementation detail.

	The loss is a scalar functional of $\{\langle Z_i\rangle_b\}$, so it is back-propagated through the state~\cite{Jones_2020_adjoint}. With $O_b=\sum_i C_{bi}Z_i$ diagonal, the costate $\ket{\lambda_b}=O_b\ket{\psi_b}$ is one elementwise product, and a single backward sweep yields every angle derivative at the cost of three state evolutions, independent of parameter count. Agreement with the parameter-shift rule, which remains the hardware-executable gradient, is $10^{-14}$. 
	%	The speedup at $n=8$, $L=4$ is $47\,\mathrm{s}\to1.7\,\mathrm{s}$ per L-BFGS iteration, and it is what brings full convergence within reach ($600$ iterations, common-random-number germ batch, unbiased $U$-statistic estimator). 
	Because the frozen-batch $U$-statistic can drift below zero by sample memorization, the fit additionally evaluates the loss on \emph{fresh} germs at fixed intervals, stops when the fresh-germ value ceases to improve, and saves the best-validation iterate rather than the last. The frozen-batch curve continues to fall by a factor of two after the fresh-germ curve has flattened, which is the memorization made visible. Every reported number is computed out of sample, fresh germs against the held-out test split.
	
	\subsection{Parameter accounting}
	\label{sec:accounting}
	
	The only quantities adjusted by any optimizer in QPCE are the $154$ gate angles $(\bm\beta,\bm\phi,\bm\theta,\bm a,\bm w)$. No classical quantity is fitted against any loss, at any stage, by any method. There is no linear solve, no regularization path, no hyperparameter selected against a validation objective beyond the early-stopping iterate itself, and no classical head of any kind. The deployed model stores the $154$ angles and the $2n=16$ readout constants, $170$ numbers in total, and Table~\ref{tab:resources} collects the complete budget.
	
	\begin{table}[t]
		\small
		\centering
		\caption{QPCE resources and model dimensions for the deployed instance ($n=8$, $L=6$, path graph, one shared germ wire). The two-qubit depth is the chromatic-index floor $L\chi'(G)$, attained exactly because the circuit is native and edge-coloured (Sec.~\ref{sec:topology}).}
		\label{tab:resources}
		\renewcommand{\arraystretch}{1.15}
		\begin{tabular}{lc}
			\toprule
			Resource & QPCE \\
			\midrule
			Qubits                              & $8$ \\
			Interferometer blocks $L$           & $6$ \\
			Coupler edges $|E|$ (path)          & $7$ \\
			$R_z$ phases $\bm\beta$ ($Ln$)      & $48$ \\
			$\RZZ$ phases $\bm\phi$ ($L|E|$)    & $42$ \\
			Wall angles $\bm\theta$ ($Ln$)      & $48$ \\
			Shared-germ amplitudes $\bm a$ ($n$)& $8$ \\
			Private-germ amplitudes $\bm w$ ($n$)& $8$ \\
			\textbf{Quantum parameters}         & $\bm{154}$ \\
			\textbf{Fitted classical parameters}& $\bm{0}$ \\
			Stored classical constants ($2n$)   & $16$ \\
			\midrule
			Germ wires, i.i.d.\ $\cU(0,1)$      & $n+1=9$ \\
			Measured observables                & $8\ \langle Z_i\rangle$ \\
			Measurement settings per sample     & $\bm{1}$ \\
			Two-qubit pulses (fractional $\RZZ$)& $42$ \\
			Two-qubit layers ($L\chi'$)         & $12$ \\
			Routing SWAPs                       & $\bm{0}$ \\
			Shots per germ (hardware)           & $1152$ \\
			Circuits per generated shower       & $\bm{1}$ \\
			\midrule
			Training samples $N_{\text{train}}$ & $2600$ \\
			Test samples $N_{\text{test}}$      & $1400$ \\
			\bottomrule
		\end{tabular}
	\end{table}
	
	% ============================================================================
	\section{Verifying that the Dependence is Generated by Entanglement}
	\label{sec:verification}
	% ============================================================================
	
	Verification is the \textsc{Verify} procedure of Algorithm~\ref{alg:qpce}, and the shipped training script runs it automatically at the end of every fit. The generated dependence has two channels placed there by design, the entangling couplers and the shared germ wire, and the protocol separates them. First it confirms that the circuit stripped of its entanglers produces no dependence at all, conditionally on the shared germ, whatever the remaining angles happen to be. Then it measures the shared-latent channel alone by releasing the wire while the couplers stay off. Whatever dependence the full model shows beyond what the released wire supplies is produced by the entangling gates. We state three propositions about the construction and turn each one into an experiment. A hybrid model of comparable scope could not pass these tests, and in most cases could not even formulate them, because its dependence has no named location inside the model.
	
	Every statistic we use is invariant under monotone per-cell response, which is what allows the same numbers to be quoted in simulation and on hardware. It is worth being exact about the reach of these fixed-basis measurements. They certify where in the circuit the generated dependence comes from, which is a statement about the machine rather than about the entanglement of the state it prepares. A single measurement basis yields one probability distribution, and a separable device reproducing the same $\langle Z_i\rangle(\bm\varepsilon,\varepsilon_0)$ would yield it too, since entangled states exist whose measurement statistics admit local hidden-variable models~\cite{Werner_1989}. Certifying the state itself calls for more than one setting~\cite{Horodecki_RMP_2009}, which is what we add in Sec.~\ref{sec:bellcap}.
	
	\subsection{The no-entanglement floor}
	
	\begin{proposition}[Entangler-free $\Rightarrow$ conditional independence]
		\label{prop:floor}
		Let $\phi_{\ell,ij}=0$ for all $\ell$ and all $(i,j)\in E$. Then for \emph{any} values of $\bm\beta,\bm\theta,\bm a,\bm w$ and any $L$, the state \eqref{eq:circuit} is a product state, and the generated intensities $Y_i=m_i+s_i\langle Z_i\rangle$ are mutually independent \emph{conditionally on the shared germ} $\varepsilon_0$. At any fixed $\varepsilon_0$ they are mutually independent random variables, and in the private-only family ($\bm a=0$) the independence is unconditional. Because no classical map follows the circuit, the conclusion applies to the entire generative output, marginals and dependence alike.
	\end{proposition}
	
	\begin{proof}
		With all $\RZZ$ removed, every gate in \eqref{eq:circuit} acts on a single qubit, so the circuit unitary factorizes as $\bigotimes_i V_i(\varepsilon_i,\varepsilon_0)$ with $V_i = \prod_{\ell}\bigl[R_y(\theta_{\ell,i})\,R_z(\beta_{\ell,i})\,R_y(\pi(w_i\varepsilon_i+a_i\varepsilon_0))\bigr]$. Acting on $\ket{0^n}$ this yields a product state, and the Born expectation of a local observable on a product state depends only on its own factor, so $\langle Z_i\rangle = g_i(\varepsilon_i,\varepsilon_0)$ for a deterministic $g_i$. Conditionally on $\varepsilon_0$ the quantities $g_i(\varepsilon_i,\varepsilon_0)$ are functions of the independent private germs alone, hence conditionally independent, with $\bm a=0$ the dependence on $\varepsilon_0$ vanishes and the independence is unconditional.
	\end{proof}
	
	\begin{corollary}[A data-computable performance floor]
		\label{cor:floor}
		Every entangler-free member of the QPCE family, for every setting of its remaining non-coupler angles and at every depth, has rank-space discrepancy at least
		\begin{equation}
			F \;\equiv\; \MMD^2\bigl(\Pind,\;\hat C_{\mathrm{data}}\bigr),
			\label{eq:floor}
		\end{equation}
		where $\hat C_{\mathrm{data}}$ is the empirical distribution of the rank-transformed data and $\Pind$ its column-shuffled counterpart, computable \emph{from the data alone}, with no quantum resource.
	\end{corollary}
	
	As a benchmark, this floor is stronger than an empirical baseline in three respects. It is provable, being not the best entangler-free model we managed to train but a bound over the entire entangler-free, frozen-wire family. It is free, since column shuffling costs nothing and requires no second training run. And no choice of the remaining angles, no depth and no readout constant can move it, because Proposition~\ref{prop:floor} holds for all of them. One bookkeeping rule is mandatory and easy to violate silently, is that the floor and every energy distance compared against it must be computed in the \emph{same coordinate}. On this dataset the z-scored floor is $\numFfloor$ and the raw-coordinate floor is $0.0067$, a factor of $20$ apart. Every number in this paper is z-scored, and the shipped scoring tool computes the floor and the model distance in one call so that the coordinates cannot drift apart.
	
	Evaluated on z-scored coordinates, split-half so that the $U$-statistic is not contaminated by duplicate rows, we measure $F = \numFfloor \pm \numFfloorSd$ from the $2600$ training rows, and all model distances are scored against the $1400$ held-out rows in the same coordinate.
	
	\subsection{Detection significance versus effect size}
	\label{sec:sigma_vs_D}
	
	A model's position relative to the floor can be quoted two ways, and only one of them is a result. The detection significance
	\begin{equation}
		\sigma \;=\; \frac{F-\MMD^2(\hat C_{\mathrm{gen}},\hat C_{\mathrm{data}})}{\mathrm{SE}(F)}
	\end{equation}
	has the shuffle standard error in its denominator, so it answers whether the model is distinguishable from independence, a question with a very low bar. The effect size
	\begin{equation}
		D \;=\; 1-\frac{\MMD^2(\hat C_{\mathrm{gen}},\hat C_{\mathrm{data}})}{F}\;\in\;[0,1]
		\label{eq:effectsize}
	\end{equation}
	answers how much of the available dependence the model captured.
	
	The distinction is not academic. Table~\ref{tab:sigma_calibration} calibrates both statistics against Gaussian models carrying a controlled fraction of the data's dependence. A model reproducing only $5\%$ of the dependence already registers at $4.6\sigma$, and $10\%$ registers at $12.5\sigma$. A many-$\sigma$ headline is therefore a statement about the tightness of the shuffle error bar, not about model quality. \textbf{We report $D$ as the result and $\sigma$ only as a detection flag.}
	
	\begin{table}[t]
		\centering
		\caption{Calibration of the two floor statistics. A Gaussian reference model carrying a controlled fraction $s$ of the data's dependence is scored against the floor. The detection significance $\sigma$ saturates almost immediately, while the effect size $D$ tracks $s$ faithfully.}
		\label{tab:sigma_calibration}
		\renewcommand{\arraystretch}{1.15}
		\begin{tabular}{ccc}
			\toprule
			dependence fraction $s$ & $\sigma$ below floor & effect size $D$ \\
			\midrule
			$0.05$ & $4.6$  & $0.059$ \\
			$0.10$ & $12.5$ & $0.162$ \\
			$0.25$ & $26.9$ & $0.348$ \\
			$0.50$ & $52.2$ & $0.675$ \\
			$1.00$ & $76.3$ & $0.987$ \\
			\bottomrule
		\end{tabular}
	\end{table}

	\subsection{Switching off the entanglers}
	\label{sec:strip}
	
	Corollary~\ref{cor:floor} bounds the dependence available to an entangler-free model of this family. We now measure what the trained model itself produces when its entanglers are switched off, in the two configurations that a shared germ wire makes available. We regenerate from the trained circuit with all $\phi_{\ell,ij}$ set to zero, leaving every other angle, the readout constants and the germ draws untouched, and report
	\begin{equation}
		\mathcal{T} \;=\; \max_{i\neq j}\bigl|\,\mathrm{corr}(u_i,u_j)\,\bigr|
	\end{equation}
	once with $\varepsilon_0$ frozen at its median, which realizes the conditioning of Proposition~\ref{prop:floor}, and once with $\varepsilon_0$ free, which isolates the shared-latent channel.
	
	\emph{Frozen shared wire.} Proposition~\ref{prop:floor} predicts $\mathcal{T}$ consistent with sampling noise. Measured on the deployed checkpoint, $\mathcal{T}=\numStripT$ (noiseless, $\numNgen$ draws) against a maximum of $28$ order statistic of the sampling floor of $\numStripFloor$. The couplers-off, wire-frozen correlations lie below the noise ceiling and every unit of dependence beyond the shared wire collapses.
	
	\emph{Free shared wire.} With the couplers off and $\varepsilon_0$ free, the shared-latent channel alone carries $\max|\rho_S|=\numStripFree$ and mean $|\rho_S|=\numStripFreeMean$, against mean $|\rho_S|=\numRho$ for the full model. This is the global mode that the shared wire is designed to supply, measured here in isolation. The honest attribution statement is therefore a decomposition. The shared wire supplies a rank-one global mode at mean $|\rho_S|=\numStripFreeMean$, and the entangling gates supply the remainder, the pair-resolved structure and the bulk of the magnitude, which vanishes identically when they are removed.
	
	Both configurations are needed. A model can beat the independence floor while drawing its dependence from correlated noise wires or from a shared classical input, and a dense germ encoding, in which every qubit reads every germ, is such a source. With the shared wire frozen we detect any source of this kind that operates through the private wires, and with the shared wire free we measure the size of the one such source the model contains by design. Dependence carried by a shared input survives the removal of the couplers, while dependence generated by entanglement disappears with them, so the two measured numbers separate the two origins.
	
	Because the \textsc{Verify} procedure reuses the generation circuit at two parameter settings, the test costs two additional generation passes, runs on any backend the model runs on, and is emitted by the training script without being requested.
	
	\subsection{Certification with measurement settings}
	\label{sec:bellcap}
	
	The measurements reported in Sec.~\ref{sec:strip} are taken in one fixed basis, and adding measurement settings extends what they certify. Consider the trained generator at a fixed germ as a conditional sampler in which each cell is measured in one of two local bases chosen per shot. For a classical generative model with local response, in which each cell output is a function of a shared latent variable and of that cell's own setting, the four resulting correlations satisfy the Clauser-Horne-Shimony-Holt inequality $|S|\le2$, whatever the architecture and however many parameters it carries. We verify this bound by enumerating all $16$ deterministic local-response strategies, whose convex hull is the set of laws available with shared randomness. We call $S$ evaluated at the optimal settings for a given germ and pair the CHSH capability of that conditional generator, and it measures how far the conditional law lies outside the reach of the classical class.
	
	Measured on the deployed checkpoint over $10^{3}$ germs and all seven coupler pairs, the trained states reach a CHSH capability of $\numBellSmax$ against the Tsirelson maximum $2\sqrt2$, and $\numBellFrac$ of germ-pair instances lie above the classical bound. A simulated finite-shot protocol at the strongest condition gives $S=\numBellShot\pm\numBellShotSe$, an exclusion of $\numBellSigma$ standard deviations, for a model trained on calorimeter data with no Bell objective anywhere in its loss. Switching the entanglers off returns pure product states whose CHSH capability equals the classical bound to machine precision, $S=2.000000$ at every germ and every pair. The couplers are therefore the source of the conditional statistics that lie beyond the classical class, which is the same conclusion we reach from the fixed-basis measurements in Sec.~\ref{sec:strip} for the dependence itself. Executing this protocol on hardware requires the deployed circuit and one setting-dependent single-qubit rotation on each tested cell, and would make the certification device independent. We report the capability here as a measured property of the trained model in simulation.
	
	\subsection{Effect of device noise on the generated dependence}
	
	\begin{proposition}[Rank dependence is invariant under per-cell monotone device response]
		\label{prop:contraction}
		Let the device return distorted expectations $\tilde m_i=h_i(m_i)$ with each $h_i$ strictly increasing on $[-1,1]$, in particular the affine gain $\tilde m_i=c_im_i+d_i$ with $c_i>0$ of depolarizing contraction with readout bias. Then every rank statistic of the generated showers (Spearman and Kendall correlations, the empirical rank distribution, the rank MMD entering $D$) is \emph{exactly unchanged}, while the marginal densities are distorted. Conversely, no per-cell calibration, however fitted, can alter a rank statistic.
	\end{proposition}
	
	\begin{proof}
		$\tilde Y_i=m_i+s_i h_i$ is a strictly increasing per-cell function of the ideal output. Ranks are invariant under strictly increasing per-coordinate maps, and every quantity named is a function of the ranks. The converse holds because a per-cell calibration is itself such a map.
	\end{proof}
	
	We use these two statements differently. The first separates device error into a monotone per-cell response, which leaves every dependence figure untouched and makes those figures calibration-free hardware observables, and a non-monotone, cross-cell, germ-dependent residual, which is the part of the error that degrades dependence. The second makes the calibration rigorous. Since a per-cell map cannot create rank dependence, calibrating the marginals leaves the verification statistics unchanged. The readout-confusion unfolding and the per-qubit gain absorption described in Sec.~\ref{sec:hardware_probes} are both strictly monotone per-cell maps, so every dependence figure in this paper is byte-identical with and without them, and the calibrated marginals then test the device's departure from linear response without any possibility of contaminating the verification.
	
	\subsection{Shot contraction of the dependence}
	
	\begin{proposition}[Predicted shot contraction]
		\label{prop:shots}
		Let $\hat m_{i}$ be the $S$-shot estimator of $m_i=\langle Z_i\rangle$ at fixed germ, so $\mathrm{Var}(\hat m_i)=(1-m_i^2)/S$, independent across germs and cells. To leading order the device-estimated correlations obey
		\begin{equation}
			\hat\rho_{ij} \;\simeq\; c_i\,c_j\,\rho_{ij},
			\qquad
			c_i=\Bigl[\,1+\tfrac{\E_{\bm\varepsilon}[1-m_i^2]}{S\,\mathrm{Var}_{\bm\varepsilon}(m_i)}\Bigr]^{-1/2},
			\label{eq:shotlaw}
		\end{equation}
		with every quantity on the right computable from the trained model \emph{before} deployment. Moreover, because shot noise is independent across cells, it can only attenuate cross-cell dependence, never increase it.
	\end{proposition}
	
	\begin{proof}
		This is regression dilution. Adding independent noise of variance $\E[1-m_i^2]/S$ to a signal of variance $\mathrm{Var}(m_i)$ multiplies each correlation by $c_ic_j$ with $c_i$ as stated. Independence across cells forbids spurious covariance.
	\end{proof}
	
	The statement is of direct experimental use. The entire shot-noise cost of the hardware protocol is predicted from the trained model alone, so after a device run the genuine hardware noise is the \emph{residual beyond} the prediction, not the whole gap. For the \texttt{ibm\_fez} run we report in Sec.~\ref{sec:hardware}, at the device's effective shot count the predicted mean pairwise attenuation is $\numHwAtten$. Deconvolving it from the measured hardware dependence gives an implied noiseless mean $|\rho_S|=\numHwRhoDeconv$ against $\numRho$ in exact simulation, so shot noise accounts for a $1.4\%$ contraction and the residual $0.02$ deficit is genuine device error, cleanly separated. The same equation sizes the shot budget before submission, and it is how the $\numHwShots$-shot setting of the deployment was chosen.
	
	% ============================================================================
	\section{Results in Noiseless Simulation}
	\label{sec:results_sim}
	% ============================================================================
	
	\subsection{Dataset and trained configuration}
	
	We use the $8$-cell CLIC calorimeter shower dataset~\cite{clic_ds}, $N_{\text{total}}=4000$ Geant4 samples with a $65/35$ train/test split. The deployed model has $L=6$ (chaos order $6$), the eight-qubit path topology with one shared germ wire, and $154$ gate angles, trained with the energy-distance loss on a frozen batch of $1500$ germs against $1400$ training showers by L-BFGS with adjoint gradients, fresh-germ validation every $25$ iterations, and early stopping on the validation curve. The saved checkpoint is the best-validation iterate. Every quoted number is evaluated on $\numNgen$ fresh germs against the $1400$ held-out showers, in z-scored coordinates, with subsampling confidence intervals.
	
	\subsection{Verification}
	
	With the couplers removed and the shared wire frozen, the generated cells are exactly independent (Proposition~\ref{prop:floor}). Measured, the frozen strip test sits at $\max|\rho|=\numStripT$ against a sampling ceiling of $\numStripFloor$, with the wire free, the shared-latent channel alone carries mean $|\rho_S|=\numStripFreeMean$ (Sec.~\ref{sec:strip}). The trained model reaches mean $|\rho_S|=\numRho$ $[\numRhoCiLo,\numRhoCiHi]$ and rank-dependence effect size $D=\numD$ $[\numDCiLo,\numDCiHi]$ against the independence floor $F=\numFfloor\pm\numFfloorSd$. The dependence beyond the shared-latent channel is therefore produced by the entangling gates, and the attribution covers the \emph{entire} generative output, marginals included, since no classical map intervenes.
	
	\subsection{Distributional metrics}
	\label{sec:metrics}
	
	Per-cell Wasserstein-$1$ distances average $\overline{W_1}=\numWoneCell$ in intensity units. The total-energy distance is $W_1(E)=\numWoneE$, which is $\numWoneEpct\%$ of the spectrum mean and $\numWoneEwidth\%$ of the spectrum width (Fig.~\ref{fig:energy}). The energy-sum width ratio $\mathrm{sd}[\sum_kY_k]/(\sum_k\mathrm{Var}[Y_k])^{1/2}$, which equals $1$ exactly for independent cells and is blind to the marginals by construction, is $\numWidthModel$ against $\numWidthData$ in data, so the model slightly overshoots the near-cancellation that energy conservation imposes on the data's mixed-sign correlations. The Spearman residual is $\lVert\Delta\rho_S\rVert_F=\numResidF$ with mean $|\rho_S|=\numRho$ against $\numRhoData$ in data (Fig.~\ref{fig:correlations}). The joint energy-distance two-sample permutation test resolves the residual, $\mathcal E=\numEmodel$ at $p=\numPermP$ ($700$ vs $700$, $500$ permutations): at these sample sizes the model is statistically distinguishable from Geant4, a detection statement that is logically separate from the effect-size statement (Sec.~\ref{sec:sigma_vs_D}), under which the model captures $D=\numD$ of the available dependence.
	
	\subsection{The dense-graph reference}
	\label{sec:densegraph}
	
	The same ansatz trained on the dense $C_8(1,2)$ graph at $L=4$, the expressivity-maximizing configuration studied in Sec.~\ref{sec:topology}, reaches $D=0.994$ $[0.941,1.010]$ and mean $|\rho_S|=0.562$ on the identical held-out protocol, with Frobenius residual $0.347$. That model is the expressivity ceiling of the family at $n=8$ and is quoted here as such. It embeds in no present device, and its routed realization costs $316$ two-qubit pulses in $98$ layers with a projected certification budget above $10^8$ shots (Sec.~\ref{sec:topology}). The deployed path model gives up $0.10$ of $D$ and buys a factor $30$ in shot budget and a circuit that exists. Whether intermediate graphs on square-lattice devices recover the difference is a measurement we have not made, and the shipped sweep tool performs it.
	
	\subsection{Dependence structure and tail behaviour}
	\label{sec:weak}
	
	The residual dependence error is structured. The magnitude deficit is four percent in the mean absolute rank correlation ($\numRho$ against $\numRhoData$), and the Frobenius residual of $\numResidF$ concentrates in the longest-range pairs, which is the expected signature of a path graph whose causal cone is $\pm L$ sites supplemented by a rank-one global mode.
	
	Tail behaviour is the regime that separates model families, and we give it its own statistic. At quantile $q$ define the empirical tail-dependence coefficients $\lambda_L(q)$ and $\lambda_U(q)$ averaged over pairs, and the asymmetry $\Lambda(q)=\lambda_L(q)-\lambda_U(q)$. On the held-out data at $q=\numTailQ$, $\lambda_L=\numLamLdata$ $[\numLamLdataCiLo,\numLamLdataCiHi]$ against $\lambda_U=\numLamUdata$ $[\numLamUdataCiLo,\numLamUdataCiHi]$: the intervals are disjoint and $\Lambda=\numLambdaData$. The model gives $\lambda_L=\numLamLmodel$, $\lambda_U=\numLamUmodel$, $\Lambda=\numLambdaModel$, reproducing the sign and $62\%$ of the size of the asymmetry at the quoted quantile, and Table~\ref{tab:tails} resolves the three most strongly coupled pairs. Section~\ref{sec:tails} shows why this statistic, not the correlation matrix, is where model families separate on this dataset.
	
	The honest asymptotic statement is a no-go, and it applies to this model class including the deployed instance.
	
	\begin{proposition}[Smooth expectation readout has no intermediate asymptotic tail dependence]
		\label{prop:tailnogo}
		Let $Y_k=m_k+s_k\,g_k(\bm\varepsilon)$ with each $g_k$ a trigonometric polynomial (hence real-analytic) on $[0,1]^d$ and $\bm\varepsilon$ absolutely continuous with bounded density. If $g_k$ attains its minimum at non-degenerate critical points, then for any pair $(i,j)$ the asymptotic lower tail-dependence coefficient satisfies $\lambda_L\in\{0,1\}$: it is $1$ if $g_i$ and $g_j$ share an argmin and $0$ otherwise, and likewise for $\lambda_U$ at the maxima.
	\end{proposition}
	
	\begin{proof}[Proof sketch]
		Near a non-degenerate minimum the sublevel set $\{g_k<\min g_k+\delta\}$ is a Morse ball of radius $\cO(\sqrt\delta)$ about the argmin. If the argmins of $g_i$ and $g_j$ differ, the generic case, the two balls are disjoint for small $\delta$, so the joint tail probability vanishes faster than $q$ and $\lambda_L=0$. If they coincide, the balls coincide to leading order and $\lambda_L=1$. Absolute continuity of $\bm\varepsilon$ rules out atoms that could produce intermediate values.
	\end{proof}
	
	The corollary is architectural. No expectation-value readout of a smooth ansatz driven by absolutely continuous latents, at any depth, on any coupler graph, with or without the shared wire, can produce asymptotic tail dependence strictly between $0$ and $1$. The finite-$q$ asymmetry the model does reproduce is real and is the quantity a calorimeter measurement can access ($q\gtrsim1/N$), but claims about the $q\to0$ limit are foreclosed for this model class, and we make none. The no-go is moreover stronger than it first appears, and its exact boundary is worth recording, because it dictates the architecture of any successor.
	
	\begin{corollary}[Tail dichotomy for branch mixtures]
		\label{cor:branchmix}
		Let the output law be a finite mixture $\sum_b w_b\,\mathcal{L}_b$ in which each branch $\mathcal{L}_b$ is generated as in Proposition~\ref{prop:tailnogo} from absolutely continuous latents through smooth maps $g^{(b)}_k$ with non-degenerate minima. Then every pairwise $\lambda_L\in\{0,1\}$ still holds whenever the branch attaining the lowest joint values has distinct argmins for the pair, so a discrete common latent that merely selects among smooth branches cannot produce intermediate asymptotic tail dependence. If instead the extreme-dominating branch is \emph{comonotone at its minimum} for the pair, all cells driven monotonically by one scalar latent near the joint minimum, then $\lambda_L$ equals the probability that the tail event falls in a comonotone branch, and every value in $(0,1)$ is attainable by weighting comonotone against smooth branches.
	\end{corollary}
	
	\begin{proof}[Proof sketch]
		For $q$ below the quantile at which the lowest-lying branch dominates each margin's tail, the joint tail probability decomposes over branches. In a smooth branch the Morse-ball argument of Proposition~\ref{prop:tailnogo} applies verbatim and contributes $o(q)$ unless argmins coincide. In a comonotone branch the pair's tail events are identical up to monotone reparameterization, so the branch contributes its full conditional tail mass, which is proportional to $q$. The ratio converges to the comonotone weight fraction.
	\end{proof}
	
	The corollary is verified numerically on the deployed checkpoint at $2\times10^5$ draws, where the deployed model's $\bar\lambda_L(q)$ falls $0.174\to0.027$ as $q$ decreases from $0.05$ to $0.002$, a Bernoulli branch selecting between two smooth parameter settings falls identically ($0.138\to0.011$), and an engineered branch of weight $0.15$, a conditional reset followed by a single shared-germ re-upload on cells $0$--$3$, comonotone and undercutting the main branch, produces a plateau at $0.222$--$0.235$ across two decades of $q$ against the parameter-free prediction $6/28=0.214$, with the within-branch pairs at $\lambda_L=1.000$ and all others at $0.013$. Operationally the escape is mid-circuit measurement with feed-forward, conditional reset plus one conditional germ rotation, native to present hardware. Building a trained generative model on it is the designated follow-up and is outside the scope of this paper.
	
	\begin{table}[t]
		\small
		\centering
		\caption{Empirical tail-dependence coefficients at $q=\numTailQ$ for the three most strongly coupled cell pairs, deployed model against Geant4 on the held-out split.}
		\label{tab:tails}
		\renewcommand{\arraystretch}{1.15}
		\begin{tabular}{ccccc}
			\toprule
			pair $(i,j)$ & $\lambda_L$ model & $\lambda_L$ data & $\lambda_U$ model & $\lambda_U$ data \\
			\midrule
			$(1,2)$ & $\numLamLabM$ & $\numLamLabD$ & $\numLamUabM$ & $\numLamUabD$ \\
			$(2,3)$ & $\numLamLbcM$ & $\numLamLbcD$ & $\numLamUbcM$ & $\numLamUbcD$ \\
			$(6,7)$ & $\numLamLghM$ & $\numLamLghD$ & $\numLamUghM$ & $\numLamUghD$ \\
			\bottomrule
		\end{tabular}
	\end{table}
	
	% ============================================================================
	\section{Results on \texttt{ibm\_fez}}
	\label{sec:hardware}
	% ============================================================================
	
	\subsection{Protocol}
	\label{sec:hardware_protocol}
	
	What runs on the device is the \textsc{Generate} procedure of Algorithm~\ref{alg:qpce}, unchanged, on \texttt{ibm\_fez}, a $156$-qubit Heron r2 processor~\cite{IBMQuantumPlatform_2024}, via the Qiskit Runtime Sampler~\cite{QiskitRuntime_2024}. The deployed circuit is built \emph{native}: the eight-qubit path is placed by exhaustive search over all $136$ zero-SWAP embeddings, scored by the summed calibrated two-qubit and readout errors of each candidate region (Sec.~\ref{sec:topology}), and the couplers execute as \emph{fractional} $\RZZ(\phi)$ pulses, one hardware pulse per logical coupler with no CNOT decomposition. The compiled instance is asserted, not hoped, to match its budget: $\numHwPulses$ two-qubit pulses in $\numHwLayers$ layers, zero SWAPs, verified by a dry-run ISA inspection before any shots are spent, and the full gate schedule, including the device's $\RZZ$ angle convention and bit ordering, is checked against the trained model by independent numerical emulation before submission, so a sign or ordering error cannot reach the device silently. Pauli gate twirling is unavailable in combination with fractional gates on this stack, a trade we accept and then test (Sec.~\ref{sec:hardware_why}). Dynamical decoupling is enabled. One production job binds $B=\numHwGerms$ germ vectors into the template at $S=\numHwShots$ shots each, one shower per germ. The germ batch is archived at submission and the collection step refuses to score results against any other germs. A validation gate on the returned counts, marginal sanity against the binomial band and the physical range $|\langle Z\rangle|\le1$ ($0$ violations in $4000$ cell estimates), must pass before any statistic is computed.
	
	\subsection{Readout calibration}
	\label{sec:hardware_probes}
	
	Two model-independent corrections are applied, and both are provably unable to alter any dependence statistic. First, readout confusion. A dedicated calibration job of $2n+2=18$ circuits, all-zeros, all-ones, and the $2n$ single-flip states, measures the per-qubit assignment matrices, and the returned counts are unfolded by the exact tensored inverse. The calibration circuits never touch the model. Second, a per-qubit gain. The surviving noise of a shallow native circuit acts on each cell, to leading order, as a linear contraction $\langle Z_k\rangle\mapsto c_k\langle Z_k\rangle$ toward the midpoint $m_k$. The gains are estimated by regressing the device cell estimates on the noiseless reference over the archived germ batch, with the binomial shot variance subtracted so that shot noise is not misread as gain loss, and are absorbed into the stored scales $s_k$. The measured per-qubit gains span $\numHwGainCiLo$ to $\numHwGainCiHi$ with mean $\numHwGain$. Because this second correction uses the simulable reference it does \emph{not} export to unsimulable sizes, and we say so. What makes it harmless here is Proposition~\ref{prop:contraction}: both corrections are strictly monotone per-cell maps, so the Spearman matrix, the effect size $D$, the strip statistics and every tail coefficient are byte-identical before and after them. Raw and corrected marginals are reported side by side, and the headline dependence numbers are calibration-free by construction, not by promise.
	
	\subsection{Results}
	\label{sec:hardware_results}
	
	$\numHwGerms$ showers, $\numHwShots$ shots each, about $1.9$ minutes of quantum time gives, against the held-out Geant4 split in the protocol described in Sec.~\ref{sec:results_sim}: rank-dependence effect size $D=\numHwD$ $[\numHwDCiLo,\numHwDCiHi]$ against the independence floor $F=\numFfloor$, mean $|\rho_S|=\numHwRho$ $[\numHwRhoCiLo,\numHwRhoCiHi]$ against $\numRho$ in noiseless simulation and $\numRhoData$ in data, and Spearman residual $\lVert\Delta\rho_S\rVert_F=\numHwResidF$ (Fig.~\ref{fig:correlations}). The device reproduces the model's tail asymmetry at the same quantile, $\lambda_L=\numHwLamL$, $\lambda_U=\numHwLamU$, $\Lambda=\numHwLambda$ against $\numLambdaModel$ in simulation and $\numLambdaData$ in data (Fig.~\ref{fig:tails}), and the energy-sum width ratio is $\numWidthHw$ against $\numWidthData$ in data. Every number in this paragraph is a rank or ratio statistic and is therefore untouched by either calibration.
	
	\subsection{Why the device reproduces the simulated dependence}
	\label{sec:hardware_why}
	
	The agreement decomposes into three channels, all three quantified. \emph{Shot noise}, the only channel that must contract correlations, is predicted by Proposition~\ref{prop:shots} before submission. The measured effective shot count per cell is $S_{\mathrm{eff}}=\numHwSeff$ against $\numHwSeffPred$ predicted from the assumed error budget, meaning the device outperformed the budget, with implied per-$\RZZ$ error $\numHwRzzErr$, and the predicted mean pairwise attenuation is $\numHwAtten$. Deconvolving it gives an implied noiseless mean $|\rho_S|=\numHwRhoDeconv$ against $\numRho$ simulated, so the genuine device cost beyond shot noise is $0.02$ in mean $|\rho_S|$, about four percent. \emph{Monotone response}, gain and readout bias, contracts marginals and provably leaves every rank statistic fixed (Proposition~\ref{prop:contraction}), which is why the correlation figures agree far more closely than the raw marginals do. \emph{Coherent error} is the channel the first two cannot absorb, and it is visible in a reproduction test comparing the device cell estimates to the noiseless reference germ-by-germ gives $\chi^2/\mathrm{dof}=\numHwChi$ against a pass threshold of $1.5$, with per-cell residual scatter of $0.036$ to $0.044$ against a shot-only expectation of $0.022$ to $0.028$. The excess is consistent with coherent over-rotation of the untwirled fractional $\RZZ$ pulses, the price of the fractional-gate/twirling exclusivity noted above. The reproduction test resolves what the aggregate statistics leave open, and places the device close to, but measurably outside, a purely stochastic error model, and the headline dependence numbers survive because rank statistics are insensitive to a smooth germ-independent bias field in a way that per-germ residuals are not, a diagnostic separation the protocol is designed to deliver. None of these channels increases the measured dependence, since shot noise is independent across cells (Proposition~\ref{prop:shots}) and monotone response is rank-invariant (Proposition~\ref{prop:contraction}), so device error can attenuate the correlations toward the predicted level but cannot inflate them, and the measured $D=\numHwD$ is a floor on what the device genuinely transports.
	
	\begin{figure*}[htb]
		\centering
		\begin{subfigure}[t]{0.245\textwidth}\includegraphics[width=\textwidth]{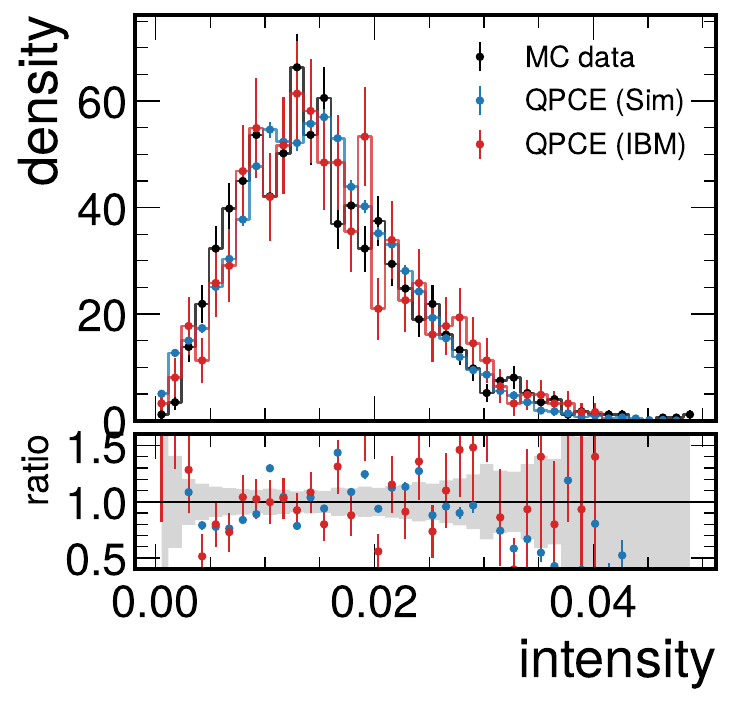}\caption*{Pixel 0}\end{subfigure}\hfill
		\begin{subfigure}[t]{0.245\textwidth}\includegraphics[width=\textwidth]{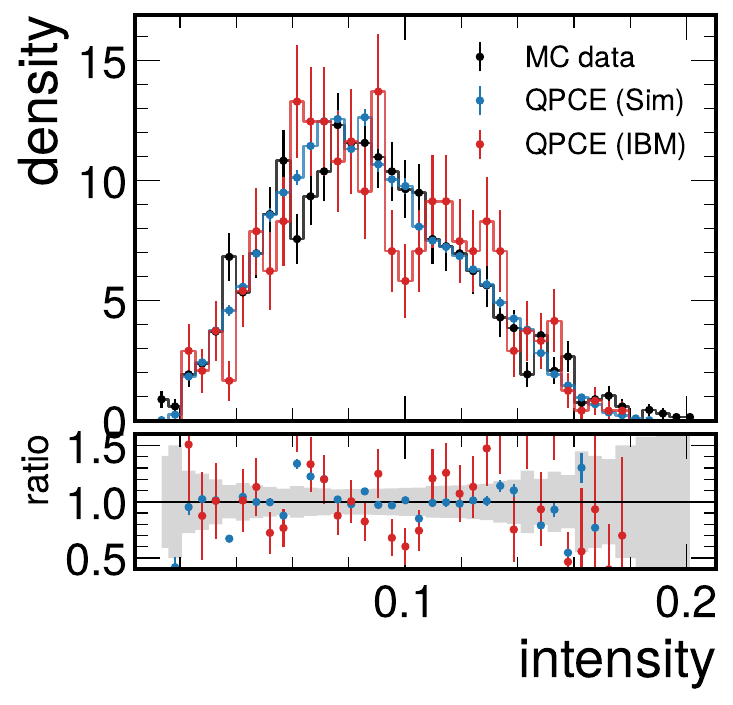}\caption*{Pixel 1}\end{subfigure}\hfill
		\begin{subfigure}[t]{0.245\textwidth}\includegraphics[width=\textwidth]{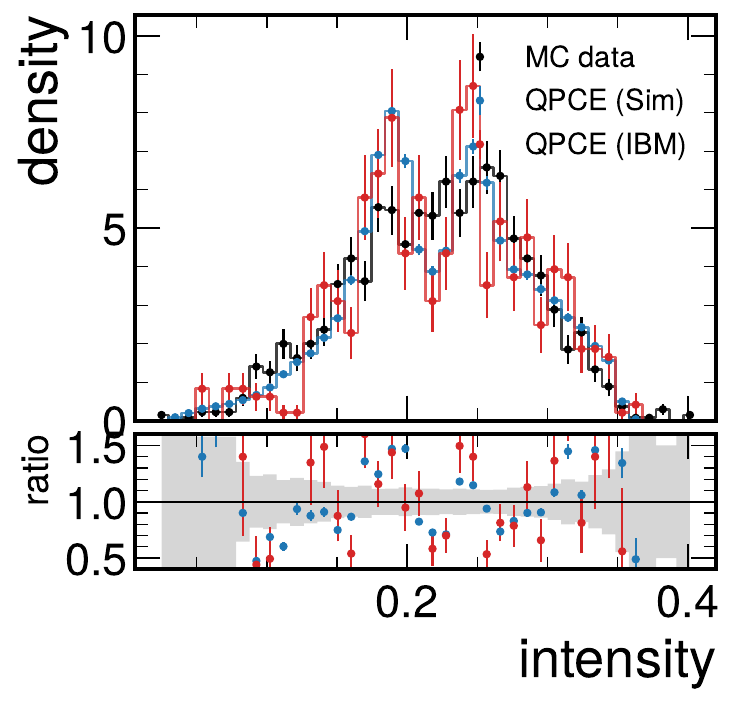}\caption*{Pixel 2}\end{subfigure}\hfill
		\begin{subfigure}[t]{0.245\textwidth}\includegraphics[width=\textwidth]{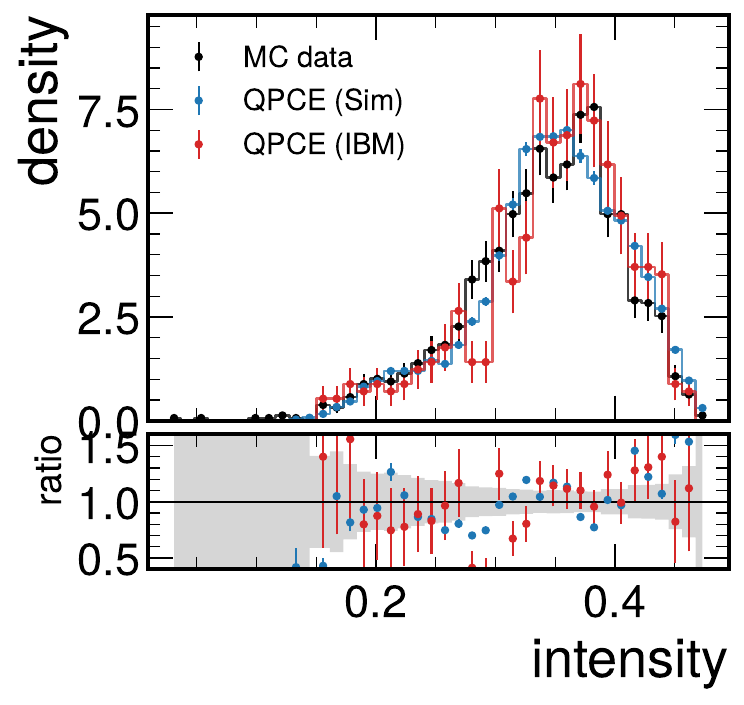}\caption*{Pixel 3}\end{subfigure}
		
		\vspace{2pt}
		
		\begin{subfigure}[t]{0.245\textwidth}\includegraphics[width=\textwidth]{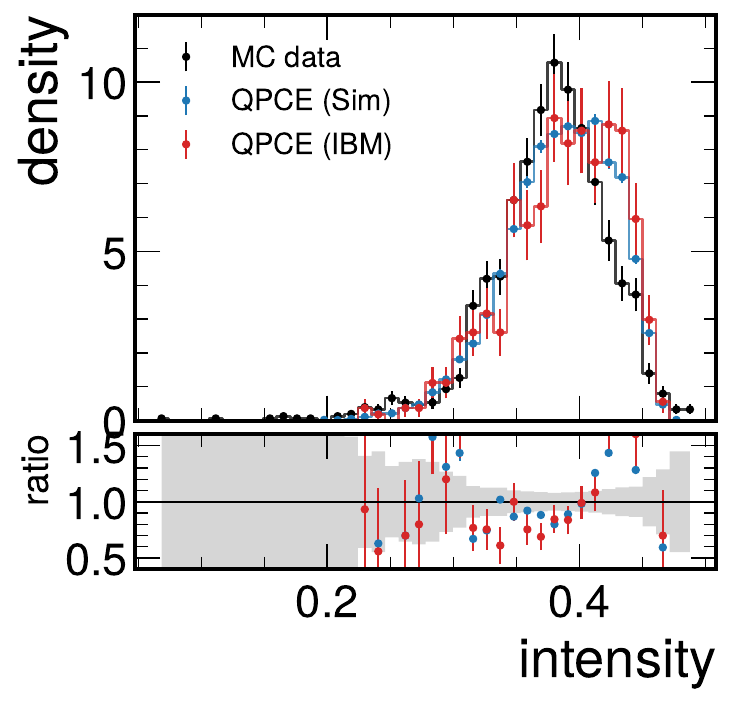}\caption*{Pixel 4}\end{subfigure}\hfill
		\begin{subfigure}[t]{0.245\textwidth}\includegraphics[width=\textwidth]{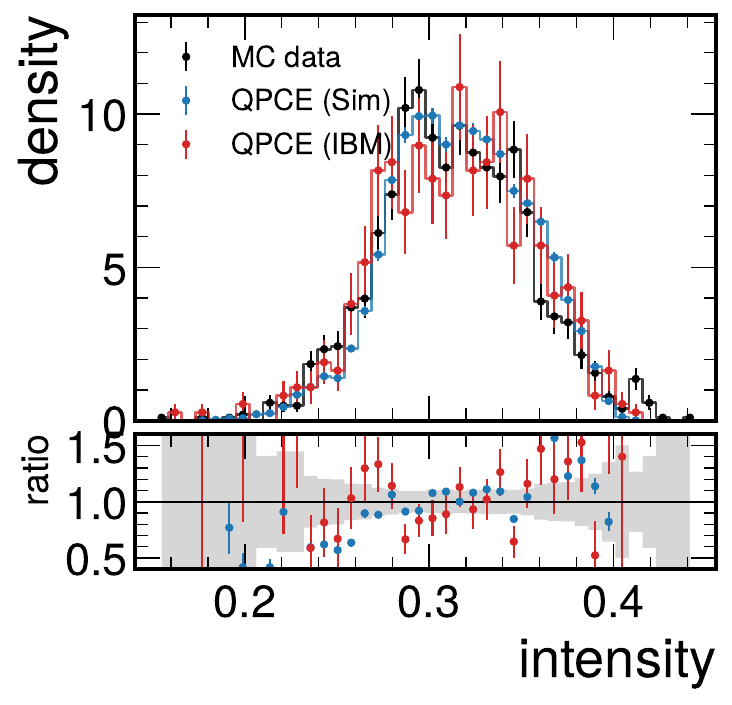}\caption*{Pixel 5}\end{subfigure}\hfill
		\begin{subfigure}[t]{0.245\textwidth}\includegraphics[width=\textwidth]{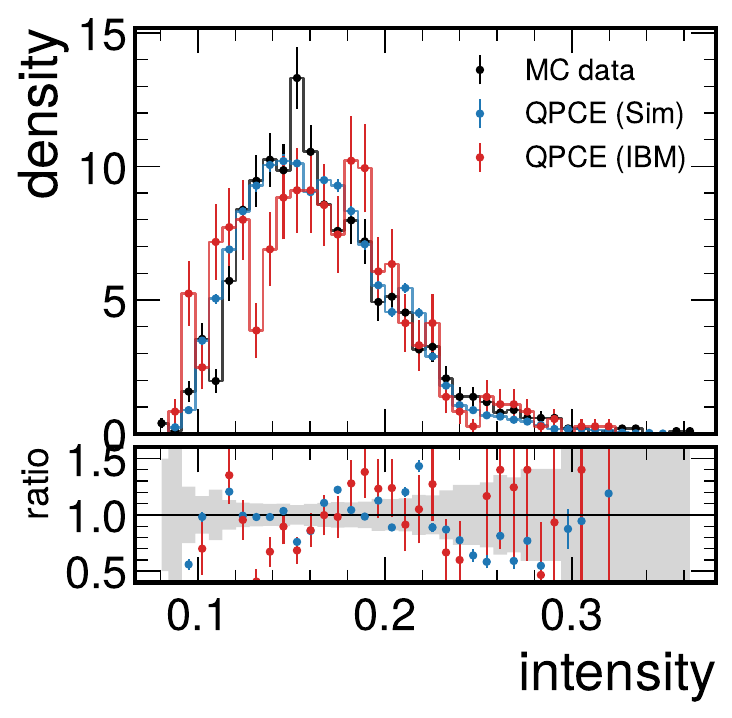}\caption*{Pixel 6}\end{subfigure}\hfill
		\begin{subfigure}[t]{0.245\textwidth}\includegraphics[width=\textwidth]{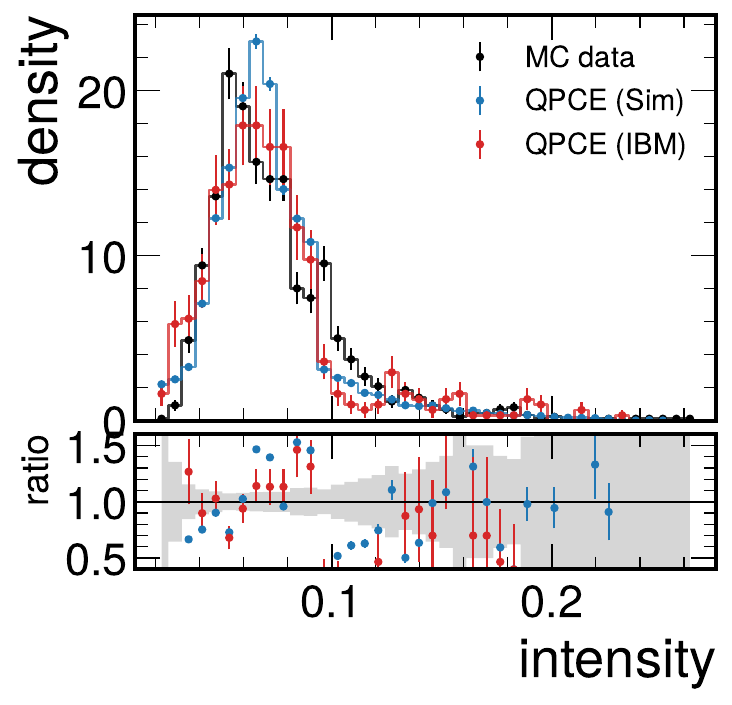}\caption*{Pixel 7}\end{subfigure}
		\caption{Marginal intensity distributions for the eight calorimeter cells. Geant4 reference (black points, $\pm1\sigma$ multinomial errors), noiseless simulation of the deployed circuit (step), and \texttt{ibm\_fez} (points), $\numHwGerms$ showers at $\numHwShots$ shots. These marginals are produced by the circuit through the linear readout, so the panels are direct device tests. The readout-confusion unfolding and the gain absorbed by the calibration in Sec.~\ref{sec:hardware_probes} affect these panels and provably nothing in Figs.~\ref{fig:correlations} and~\ref{fig:tails}. Ratio-panel and error conventions as in Appendix~\ref{app:validation}.}
		\label{fig:marginals}
	\end{figure*}
	
	\begin{figure}[tb]
		\centering
		\includegraphics[width=0.9\columnwidth]{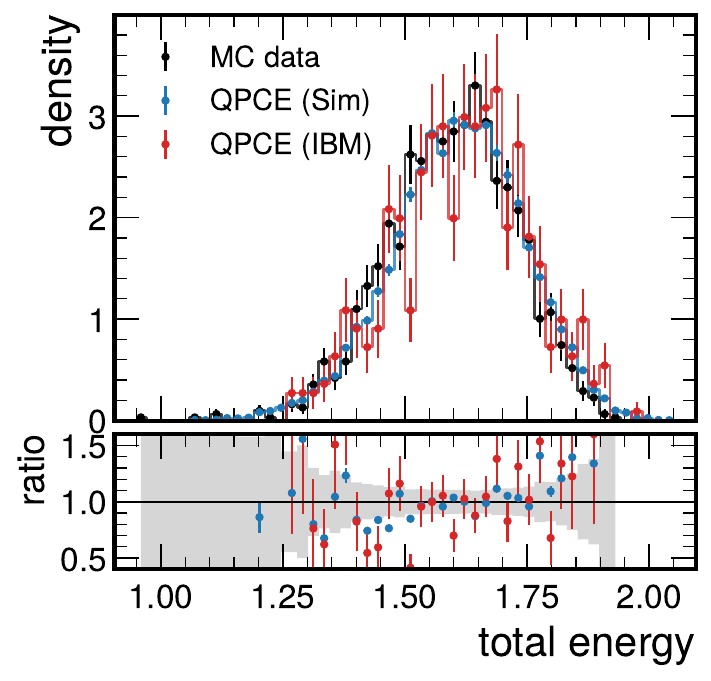}
		\caption{Total deposited energy $E=\sum_j Y_j$, conventions as in Fig.~\ref{fig:marginals}. Because nothing pins the marginals, the summed spectrum is a genuine joint test at first order, in which cell errors compound in the mean and dependence errors in the width. In simulation $W_1(E)$ is $\numWoneEpct\%$ of the spectrum mean and $\numWoneEwidth\%$ of the spectrum width. The width ratios are $\numWidthData$ (data), $\numWidthModel$ (simulation), $\numWidthHw$ (\texttt{ibm\_fez}).}
		\label{fig:energy}
	\end{figure}
	
	\begin{figure*}[tb]
		\centering
		\includegraphics[width=0.32\textwidth]{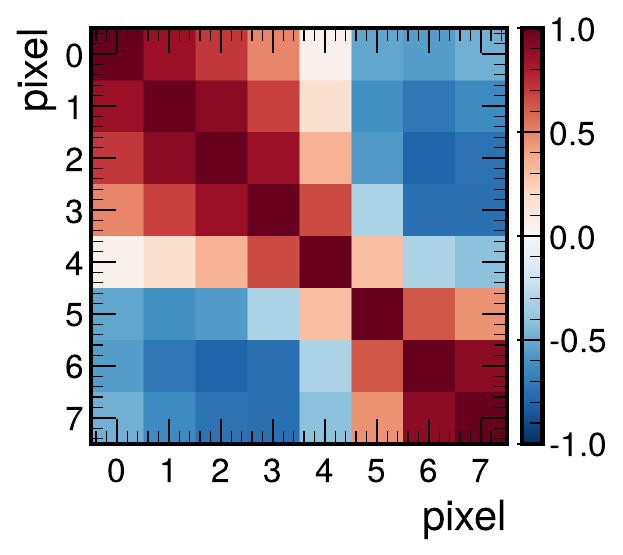}
		\includegraphics[width=0.32\textwidth]{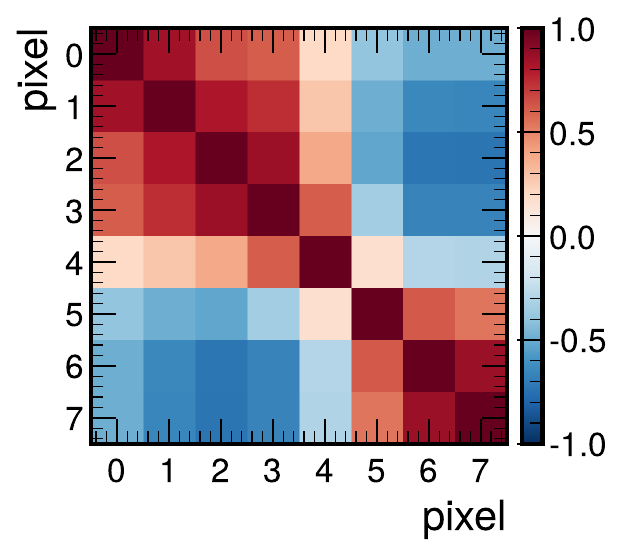}
		\includegraphics[width=0.32\textwidth]{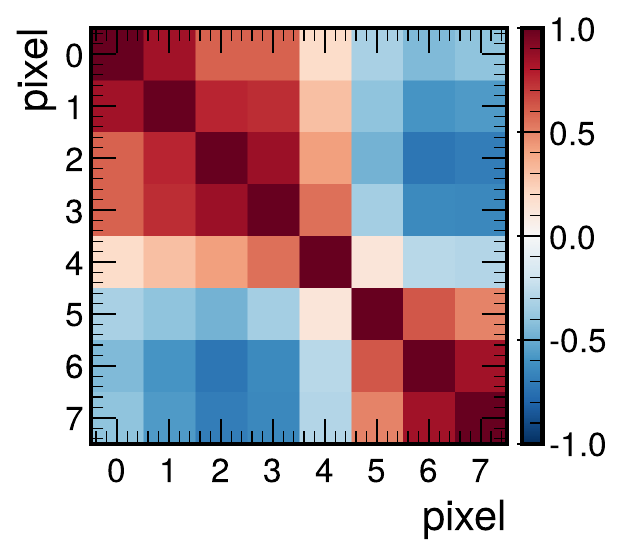}
		\caption{Spearman rank-correlation structure. Geant4 reference, deployed model in noiseless simulation, and \texttt{ibm\_fez}, on one shared symmetric scale. By Proposition~\ref{prop:contraction} this figure is invariant under any per-cell device response and any calibration. It is the calibration-free hardware observable, and the only figure whose hardware deficit is attributable to non-monotone device error. Mean $|\rho_S|$: $\numRhoData$ (data), $\numRho$ (simulation), $\numHwRho$ (device). The shot-deconvolved device value is $\numHwRhoDeconv$.}
		\label{fig:correlations}
	\end{figure*}
	
	\begin{figure*}[htb]
		\centering
		\begin{subfigure}[t]{0.245\textwidth}\includegraphics[width=\textwidth]{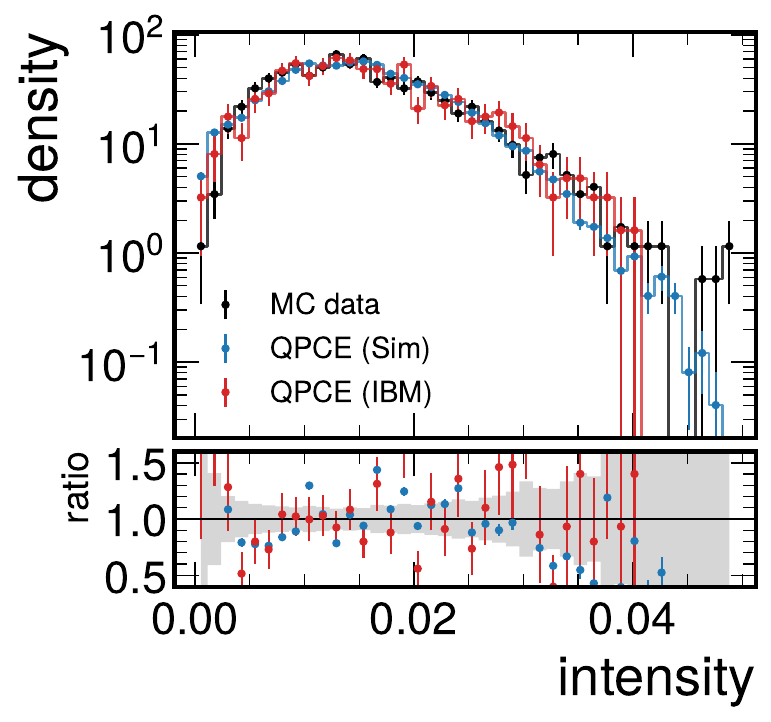}\caption*{Pixel 0}\end{subfigure}\hfill
		\begin{subfigure}[t]{0.245\textwidth}\includegraphics[width=\textwidth]{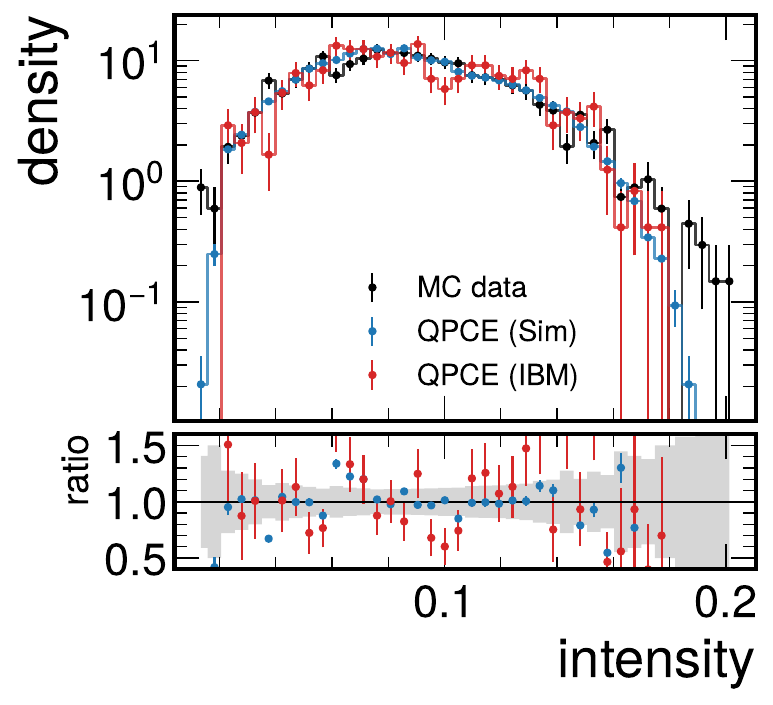}\caption*{Pixel 1}\end{subfigure}\hfill
		\begin{subfigure}[t]{0.245\textwidth}\includegraphics[width=\textwidth]{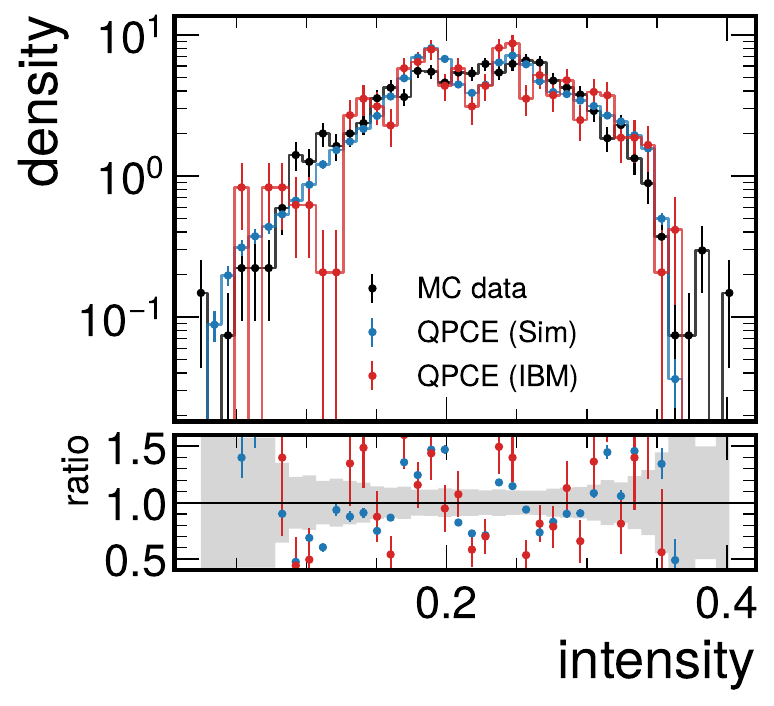}\caption*{Pixel 2}\end{subfigure}\hfill
		\begin{subfigure}[t]{0.245\textwidth}\includegraphics[width=\textwidth]{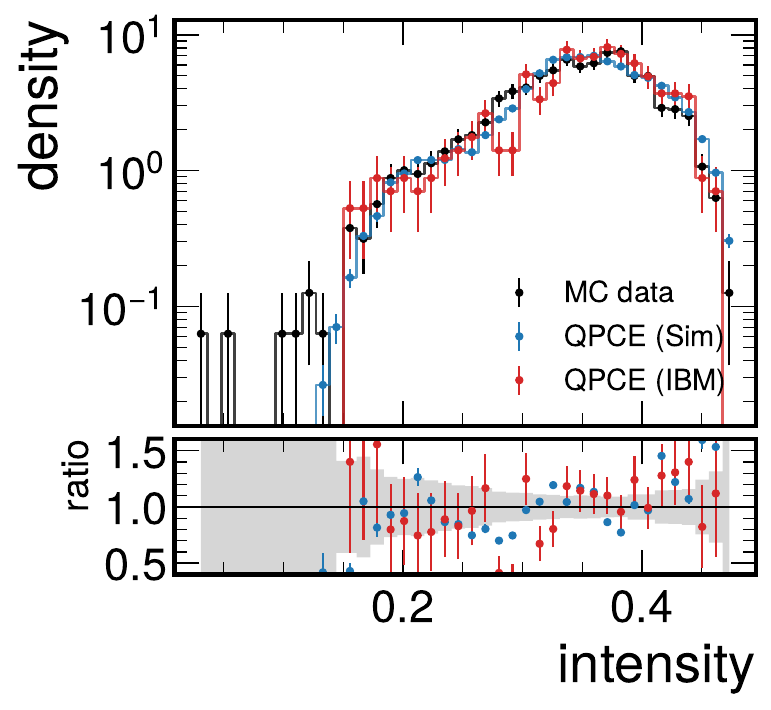}\caption*{Pixel 3}\end{subfigure}
		
		\vspace{2pt}
		
		\begin{subfigure}[t]{0.245\textwidth}\includegraphics[width=\textwidth]{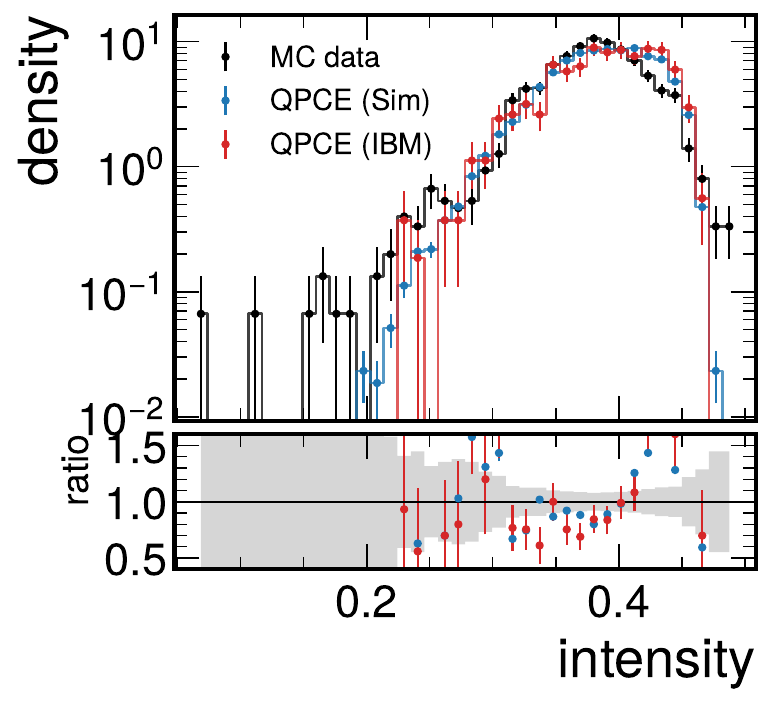}\caption*{Pixel 4}\end{subfigure}\hfill
		\begin{subfigure}[t]{0.245\textwidth}\includegraphics[width=\textwidth]{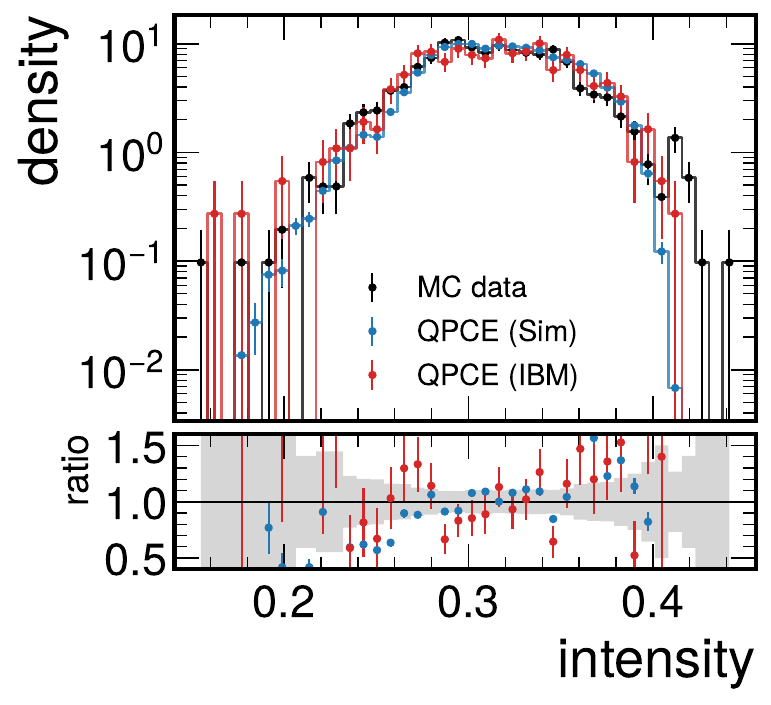}\caption*{Pixel 5}\end{subfigure}\hfill
		\begin{subfigure}[t]{0.245\textwidth}\includegraphics[width=\textwidth]{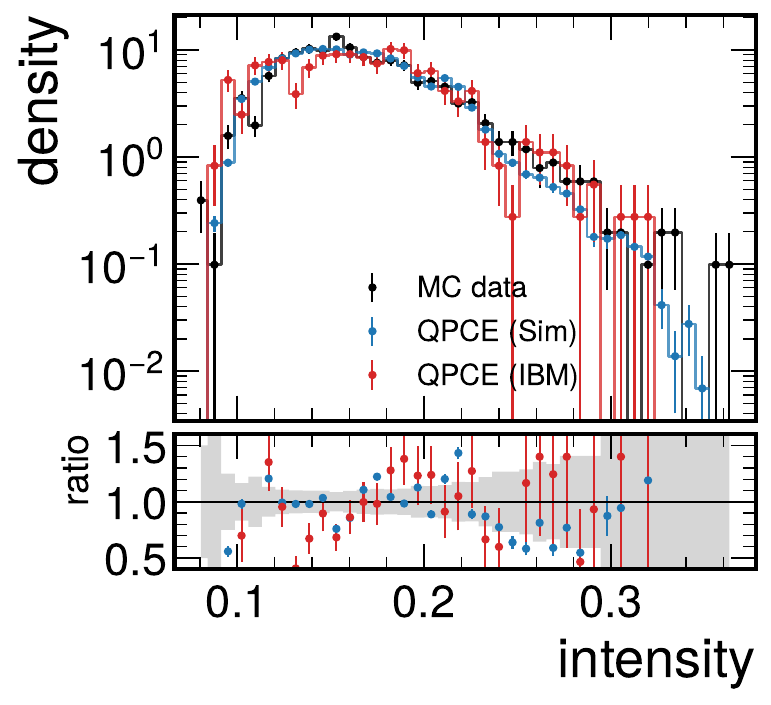}\caption*{Pixel 6}\end{subfigure}\hfill
		\begin{subfigure}[t]{0.245\textwidth}\includegraphics[width=\textwidth]{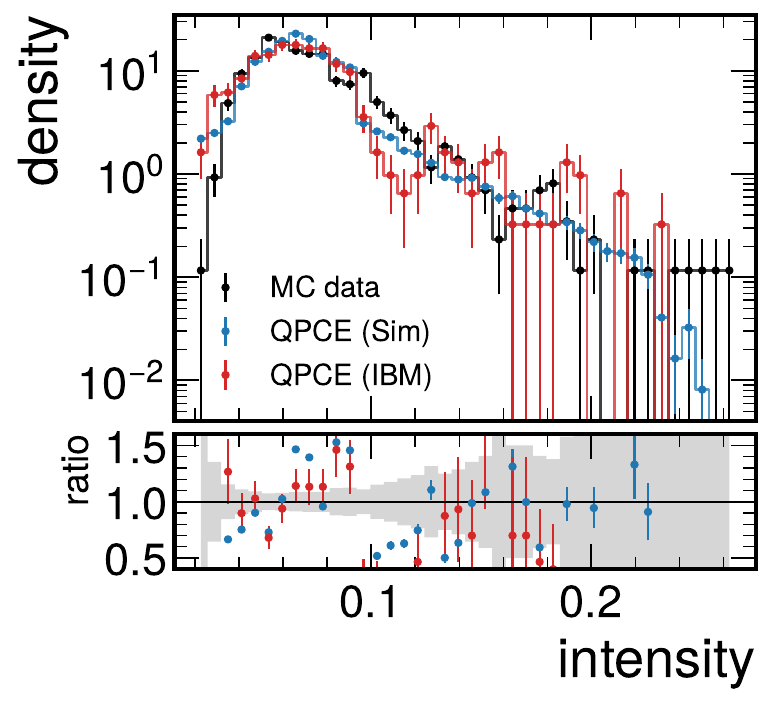}\caption*{Pixel 7}\end{subfigure}
		
		\caption{Tail behaviour, the regime where model families separate (Sec.~\ref{sec:tails}). Top rows show the per-cell marginal tails on a logarithmic scale, Geant4 against the deployed model in simulation and on \texttt{ibm\_fez}. The bottom panel shows the upper tail of the total-energy spectrum. The tail-dependence asymmetry $\Lambda=\lambda_L-\lambda_U$ at $q=\numTailQ$ is $\numLambdaData$ in data, $\numLambdaModel$ in simulation, and $\numHwLambda$ on the device. Every elliptical reference fixes it at zero identically.}
		\label{fig:tails}
	\end{figure*}
	
	% ============================================================================
	\section{Classical Consistency References and the Structure of the Target}
	\label{sec:honest}
	% ============================================================================
	
	\subsection{The Gaussian reference as a correctness check}
	
	Classical reference densities enter this paper with one purpose, to test whether the statistics QPCE produces are \emph{correct in scale}, not to stage a competition. The natural reference for dependence is the maximum-likelihood multivariate normal fitted on rank-transformed data, carrying the $n(n-1)/2=28$ entries of a correlation matrix. Because the target is near-elliptical at second order (Sec.~\ref{sec:elliptical}), this reference is close to \emph{correctly specified} there, and a correctly specified model marks the level a well-calibrated fit of this dataset must reach on second-order-dominated metrics~\cite{Embrechts_2002}. Table~\ref{tab:baselines} therefore fixes the scale of every metric with three anchors, the provable independence floor below, the correctly specified second-order reference above, and the data itself, and places the deployed circuit and its \texttt{ibm\_fez} execution between them. Read as the consistency check it is, the table confirms that the circuit's dependence statistics sit in the correct regime, $D=\numD$ between floor and reference, mean $|\rho_S|$ within seven percent of the reference's $\numGrho$, and that the one statistic on which the circuit and the reference differ in kind, the tail asymmetry $\Lambda$, differs in exactly the direction the structure of the two model families dictates (Sec.~\ref{sec:tails}).
	
	\begin{table*}[tb]
		\small
		\centering
		\caption{Consistency references on the same held-out data and metrics, tail coefficients at $q=\numTailQ$. The independence row is the provable verification floor. The Gaussian reference on ranks is close to correctly specified for the second-order structure of this near-elliptical target and anchors the attainable scale of the aggregate metrics, and $\Lambda=\lambda_L-\lambda_U$ is the tail-asymmetry statistic that every elliptical family fixes at zero identically.}
		\label{tab:baselines}
		\renewcommand{\arraystretch}{1.4}
		
		\begin{tabular}{lcccccc}
			\toprule
			Model & par. & $\lVert\Delta\bm\rho_S\rVert_F$ & $D$ & $\lambda_L$ err. & $\lambda_U$ err. & $\Lambda$ \\
			\midrule
			Gaussian reference & $28$ & $\numGresidF$ & $\numGD$ & $\numGlamLerr$ & $\numGlamUerr$ & $\numGLambda$ \\
			QPCE (simulation)  & $154$ & $\numResidF$ & $\numD$ & $\numLamLerr$ & $\numLamUerr$ & $\numLambdaModel$ \\
			QPCE (\texttt{ibm\_fez}) & $154$ & $\numHwFrob$ & $\numHwD$ & --- & --- & $\numHwLambda$ \\
			Independence       & $0$  & $\numIresidF$ & $\numID$ & $\numIlamLerr$ & $\numIlamUerr$ & $0$ \\
			\bottomrule
		\end{tabular}
	\end{table*}
	
	\subsection{Second-order structure and tail behaviour of the target}
	\label{sec:elliptical}
	
	At second order the dataset is close to elliptical. For any elliptical copula Kendall's $\tau$ fixes the dispersion parameter exactly through $\tau=(2/\pi)\arcsin\rho$~\cite{Lindskog_2003}, and for the Gaussian member the rank correlation then follows as $\rho_S=(6/\pi)\arcsin(\rho/2)$. Estimating $\rho$ from $\tau$ pair by pair and comparing the prediction with the measured $\rho_S$, we find agreement across all $28$ pairs to within $\numEllipMaxDev$, so a single correlation matrix captures the second-order structure almost completely, which is what qualifies the Gaussian reference as the correctly specified anchor of Table~\ref{tab:baselines}, and which also explains why strong second-order fits by hybrid models on data of this kind are evidence about their classical networks rather than their circuits (the attribution problem we described in Sec.~\ref{sec:intro}, in its most concrete form).
	
	The tails tell the opposite story, and it is a theorem, not a fit. Every elliptical copula is radially symmetric, so its lower and upper tail-dependence functions \emph{coincide}: $\Lambda(q)=\lambda_L(q)-\lambda_U(q)\equiv0$ identically, for the Gaussian, for every Student-$t$, for every elliptical model at any correlation matrix and any tail index~\cite{Embrechts_2002,Joe_2014}. The measured data give $\Lambda=\numLambdaData$ at $q=\numTailQ$ with disjoint bootstrap intervals on $\lambda_L$ and $\lambda_U$ (Sec.~\ref{sec:weak}). The dataset is therefore \emph{decisively non-elliptical} in exactly one measured place, and any elliptical reference, however many parameters it is given, is misspecified there by construction. This is the discriminating statistic of this dataset and a finding of this work, since any elliptical description of these showers, at any parameter count, is misspecified in the tails by construction, while the circuit family is not.
	
	\subsection{Tail asymmetry of the generated showers}
	\label{sec:tails}
	
	Against that theorem the measured behaviour reads as follows. The Gaussian reference, tail-symmetric by construction, splits the data's asymmetry down the middle, overshooting the upper tail and undershooting the lower by the same amount, $\Lambda_G=\numGLambda$. The circuit, trained on a single joint metric that was never told about tails, produces $\Lambda=\numLambdaModel$ against $\numLambdaData$ in data, the correct sign and $62\%$ of the size, and the \texttt{ibm\_fez} execution transports it at $\Lambda=\numHwLambda$. Per side, the circuit reproduces the upper tail to $\numLamUerr$ and the lower, heavier tail to $\numLamLerr$, the lower tail being the harder side for every model family considered. Archimedean alternatives do not resolve the tension, because a Clayton copula matched to the data's lower tail reaches $\Lambda\approx0.2$ but collapses the rank correlation to $0.24$ against the required $\numRhoData$, so representing this dataset's asymmetry \emph{and} its correlation level simultaneously is an open modelling problem for the standard parametric families, and the finite-$q$ asymmetry is the quantity to target, because Proposition~\ref{prop:tailnogo} forecloses the asymptotic version for this model class and the calorimeter sample sizes cap the accessible quantile at $q\gtrsim1/N$ in any case. Targets with stronger tail asymmetry are the regime where the phase-parameterized interferometer has the most to exploit, and the verification protocol travels with the model there unchanged.
	
	\section{Scaling and the Advantage Question}
	\label{sec:discussion}
	% ============================================================================
	
	\subsection{Scaling of the deployed model}
	
	Every deployed resource grows linearly in the number of cells. A segmentation of $n$ cells maps to an $n$-qubit path, which embeds with zero SWAPs on the heavy-hex lattice for any $n$ up to the device size, and the two-qubit depth floor $L\chi'=2L$ does not depend on $n$ at all. The angle count is $(2n+|E|)L+2n=\cO(nL)$. One measurement setting suffices at every size, and the germs enter as rotation angles, so no state-preparation oracle appears whose cost could grow.
	
	The dependence architecture scales by the factor-model logic that motivated it. With $K$ shared germ wires the model supplies a rank-$K$ global dependence for $(K{+}1)n$ single-qubit rotation coefficients and no two-qubit gates at all, while the path couplers carry the short-range remainder. On this dataset the measured eigenstructure tells us that $K$ is small and fixed by the physics rather than by $n$. The leading mode alone carries $64\%$ of the dependence and reconstructs the off-diagonal correlations at $0.970$, and rank two reaches $0.994$. The global channel should therefore scale with the number of collective shower modes rather than with the cell count. The statistical protocol transfers unchanged, since the independence floor, the entangler-stripping measurements and every rank statistic are sample statistics available at any $n$, and the shot law of Proposition~\ref{prop:shots} is stated per pair and is size independent.
	
	Training has a structure worth stating plainly. At fixed depth the model is strictly local. On the path, $\langle Z_i\rangle$ depends only on the $2L{+}1$ private germs inside its causal cone and on the shared wire, which we verify directly on the deployed code, where changing a germ at graph distance $7>L$ moves $\langle Z_0\rangle$ by at most $8\times10^{-16}$. The exact reference, the adjoint training gradient and the gain calibration therefore cost $\cO(n\,2^{\min(n,\,2L+1)})$, which is linear in $n$ at fixed $L$. Classical training of the deployed family extends to hundreds of cells at $L\le6$ with no approximation anywhere. The regime that leaves classical reach is the one in which the causal cone covers a large register, either because depth grows with $n$ or because the coupler graph gains a dimension. There, training would move to the parameter-shift rule on hardware at $2\times\cO(nL)$ circuits per gradient, and the reference-based gain fit would give way to a reference-free calibration. The readout-confusion unfolding and every rank statistic already need no reference.
	
	\subsection{The advantage question as a requirements map}
	
	We claim no \emph{computational} quantum advantage. The construction earns its place by attribution, verifiability and deployability, together with a separation of a different kind established at the end of this section. The scaling analysis makes the structural position exact, and it should be stated once without hedging. At fixed depth on a one-dimensional native graph, every quantity the model touches, its outputs, its gradients and its references alike, is classically computable in time linear in $n$. Growing the cell count alone cannot make this family classically hard, and a claim to the contrary would be false.
	
	Three ingredients would move a QPCE-type generator toward a classically hard regime, and each can be named precisely. The first is a coupler graph of dimension two or higher with depth growing with system size, which is the setting of the sampling-hardness results for shallow circuits. The second is a sampling readout, single-shot Born outcomes in place of expectation values, since expectation values to fixed additive precision are both shot limited on hardware and often classically approximable. The third is a branch structure from mid-circuit measurement with feed-forward in which at least one branch is comonotone at the extremes. A discrete latent that merely mixes smooth branches is provably insufficient, and the comonotone form is precisely what Corollary~\ref{cor:branchmix} identifies as the requirement for asymptotic tail dependence.
	
	Each ingredient costs something the present construction certifies, the exact reference, the exact floor, and the expectation-level attribution respectively. Within this family, computational hardness and complete verifiability pull against each other, and we regard mapping that trade as the more useful contribution than asserting either endpoint.
	
	One separation survives all three trades, because it is not computational. At the level of the conditional laws the generator can reach, the trained model already produces per-shot statistics that no classical generative model of any size can produce, $S=\numBellShot\pm\numBellShotSe$ against the universal classical bound of $2$ (Sec.~\ref{sec:bellcap}). This holds unconditionally and without leaving the classically simulable regime. Mapping the computational frontier, and certifying this reachability separation on hardware, are the two follow-ups the architecture was built to support.
	
	\section{Conclusion}
	\label{sec:conclusion}
	% ============================================================================
	
	QPCE generates calorimeter shower images, the densities and the correlations together, from $154$ gate angles, no fitted classical parameters, and a linear readout holding $2n$ stored unit conversions. The result that makes this possible is that re-uploading the germs turns circuit depth into the order of the chaos expansion, exactly degree $L$ in every germ wire and verified to machine precision. The circuit is the expansion, and nothing classical stands before it or after it.
	
	The deployed instance was designed for the machine it runs on. Seven couplers on a path reach the chromatic-index depth floor, and one shared germ wire supplies the rank-one global mode of the data through single-qubit rotations alone. Its verification splits the generated dependence into two measurements. With the couplers off and the shared wire frozen, the dependence falls to sampling noise exactly. With the couplers off and the wire released, the shared-latent channel is measured on its own. The difference between them is what the entangling gates produce. On the held-out split the model reaches $D=\numD$ against the provable independence floor and mean $|\rho_S|=\numRho$ against $\numRhoData$ in Geant4, and at these sample sizes a joint permutation test still resolves the small remaining residual at $p=\numPermP$, which is a statement about detection rather than about effect size.
	
	What we report on hardware is a deployment protocol with a measured error budget rather than a single successful run. The identical circuit on \texttt{ibm\_fez}, $\numHwPulses$ fractional $\RZZ$ pulses in $\numHwLayers$ layers with zero SWAPs, reaches $D=\numHwD$ and mean $|\rho_S|=\numHwRho$. The shot contraction was predicted before submission, the residual device cost is isolated at $0.02$ in mean $|\rho_S|$ after deconvolution, every applied correction is provably unable to alter a rank statistic, and a reproduction test resolves the one channel the corrections cannot absorb, coherent over-rotation of the untwirled fractional gates, at $\chi^2/\mathrm{dof}=\numHwChi$. The device also transports the model's tail asymmetry, $\Lambda=\numHwLambda$ against $\numLambdaModel$ in simulation.
	
	One finding will outlive this particular circuit, and it concerns the target rather than the model. The dataset is close to elliptical at second order, which is what qualifies the rank-Gaussian reference as a correctness anchor for the aggregate metrics, and it is decisively non-elliptical in its tails, $\Lambda=\numLambdaData$ with disjoint intervals, a quantity every elliptical family fixes at zero identically. The circuit reproduces the sign and most of the magnitude of that asymmetry at finite quantile. Proposition~\ref{prop:tailnogo} shows that no smooth expectation-value readout of absolutely continuous latents can reproduce it asymptotically, which turns the next step from a matter of tuning into a specific architectural requirement, a discrete common latent of the kind that mid-circuit measurement with feed-forward provides.
	
	At $n=8$ the model is exactly simulable, and we present it as such. The claim we do make is different in kind. A generative model whose samples come entirely from a circuit can carry a dependence structure rich enough for real detector data, can be deployed unchanged on present hardware, and can certify by measurement, not by argument, which of its own gates produced that dependence.
	
	\section*{Data and Code Availability}
	% ============================================================================
	The $8$-cell CLIC calorimeter shower dataset is openly available on Zenodo~\cite{clic_ds}. The complete implementation, including the trained checkpoint, the verification protocol and the scoring tool that computes the independence floor and every metric in one call, is available at~\cite{qpce_code}. 
	
\section{Author Contributions}
J.S. designed the model, implemented the simulation and training code, ran the numerical experiments and the hardware deployment, and wrote the manuscript. S.M. and F.R. contributed to the calorimeter dataset and its preparation. D.K., F.G. and K.B. supervised the project. All authors discussed the results and reviewed the manuscript.

LLM was used for editing and rephrasing of the manuscript text, for literature search and verification of citations, for restructuring and documenting the accompanying software. The results were independently verified by the authors against the shipped code.

	\begin{acknowledgments}
		This research was supported in part through the Maxwell computational resources operated at Deutsches Elektronen-Synchrotron DESY (Hamburg, Germany), a member of the Helmholtz Association HGF. The authors acknowledge support from the Helmholtz Association HGF (Germany), Hamburgische Investitions- und F\"orderbank (IFB) (Germany), the European Union's HORIZON MSCA Doctoral Networks programme project ENGAGE (101034267), and the Ministry of Science, Research and Culture of the State of Brandenburg within the Center for Quantum Technologies and Applications (CQTA) (Germany).
	\end{acknowledgments}
	
	\bibliography{references}
	
	\appendix
	
	\section{Validation of the Numerical Pipeline}
	\label{app:validation}
	% ============================================================================
	
	Every result rests on a batched statevector engine, on histogram error bars, and on figures drawn from computed numbers. Each layer is validated independently and each check is shipped with the code.
	
	\emph{Engine.} The batched statevector simulator is compared against an independent dense Kronecker implementation using $\cO(4^n)$ explicit operators. Agreement is at $3\times10^{-16}$ with an assertion threshold of $10^{-10}$, and failure aborts the pipeline.
	
	\emph{Error bars.} For fixed bin edges, row-resampling bootstrap of a histogram is identical in distribution to a $\mathrm{Multinomial}(N,\hat p)$ redraw, so
	$\mathrm{SE}(\hat f_k)=\sqrt{\hat p_k(1-\hat p_k)/N}/h_k$
	is the exact infinite-resample limit of the bootstrap, with no Monte-Carlo noise and no runtime cost. This identity is verified numerically. Correlation-matrix standard errors have no closed form and use the bootstrap directly, verified against the asymptotic $(1-\rho^2)/\sqrt N$.
	
	\emph{Ratio panels.} Following the LHC convention, the ratio points carry the reference error only, $\sigma_r=\mathrm{se}_{\text{ref}}/f_{\text{gen}}$, and the band about unity carries the model error, $1\pm\mathrm{se}_{\text{gen}}/f_{\text{gen}}$. Placing both on the same object would double-count. The delta-method ratio error is verified against the bootstrap.
	
	\emph{Goodness of fit.} The multinomial constraint makes bins negatively correlated, so the quoted $\chi^2$ is diagonal-only with $\mathrm{ndf}=n_{\text{used}}-1$, and its null calibration is verified empirically at $\langle\chi^2/\mathrm{ndf}\rangle=1.02$ over $400$ trials. Bins with $n<10$ are dropped from both the ratio panel and the $\chi^2$, with the count annotated on the figure. Binning is fixed \emph{a priori}, because $\chi^2/\mathrm{ndf}$ is binning-dependent and tuning bins until a cell passes is p-hacking.
	
	\emph{Figures.} A structural audit rebuilds each figure in memory, walks the matplotlib artist tree, and asserts that the numbers \emph{drawn} equal the numbers \emph{computed}, covering ratio point positions, error-bar half-lengths, band vertices, the $3$ to $1$ panel geometry, inward ticks on all four sides, true-vector PDF output, and absence of overlapping text, $34$ checks in total.
	
	% ============================================================================
	\section{Practical Notes on the Adjoint Training}
	\label{app:training}
	% ============================================================================
	
	Three implementation details determine whether the fit we describe in Sec.~\ref{sec:loss} converges.
	
	\paragraph*{Common random numbers.}
	
	The loss is evaluated on a frozen germ batch, making it a deterministic function of the angles. A redrawn batch would hand L-BFGS a stochastic objective and break its line search. The unbiased $U$-statistic of a frozen batch can drift below zero late in training by sample memorization, expected behaviour of an unbiased estimator near $P=Q$ rather than a defect, which is why every reported quantity is recomputed on fresh germs against the held-out split.
	
	\paragraph*{Cost of the adjoint.}
	
	For any loss that touches the state through $\{\langle Z_i\rangle\}$, one backward sweep with the diagonal costate returns all derivatives, $154$ in the deployed instance, in three state evolutions independent of the parameter count. Agreement with the parameter-shift rule is at $10^{-14}$. The shift rule is retained as the hardware-executable gradient and the adjoint as the training gradient. The measured iteration cost at $n=8$, $L=4$, batch $1500$ falls from $47\,$s with parameter shift to $1.7\,$s.
	
	\paragraph*{Convergence is a result, not a preference.}
	
	The dependence contribution to the loss is two orders of magnitude below the initial marginal mismatch, so correlations move only in the final decade of the descent. Any fixed iteration budget that terminates inside the marginal-fitting phase reads, falsely, as an expressivity failure of the ansatz. we quantify this in Sec.~\ref{sec:loss} with the capacity probe.
	
\end{document}